\documentclass[envcountsame]{llncs}
\usepackage{algorithm}
\usepackage{ntabbing}
\usepackage{etoolbox,fp}

\usepackage{amssymb,amsmath}

\usepackage{graphicx}
\usepackage{color}
\usepackage{tikz}
\usetikzlibrary{shapes,arrows,decorations.markings}

\usepackage{enumitem}

\newtheorem{thm}{Theorem}
\newtheorem{lem}[thm]{Lemma}
\newtheorem{corol}[thm]{Corollary}
\newtheorem{propo}[thm]{Proposition}
\newtheorem{defi}[thm]{Definition}

\newcommand{\PF}{\noindent {\bf Proof: }}
\newcommand{\PFof}{\noindent {\bf Proof of }}

\newcommand{\QED}{\mbox{}\hspace*{\fill}{$\Box$}\medskip}

\makeatletter
\newlength\tdima
\newcommand\tabright[1]{
      \setlength\tdima{\linewidth}
      \addtolength\tdima{\@totalleftmargin}
      \addtolength\tdima{-\dimen\@curtab}
      \makebox[\tdima][r]{#1}}
\makeatother

\def\algcomm#1{\>\tabright{{\em // #1}}}
\def\algincomm#1{{\em // #1}}

\newcommand{\refQ}{S_{\mbox{\footnotesize refine}}}
\newcommand{\RefCol}{r}
\newcommand{\AdjCol}{\mbox{Colors}_{\mbox{\footnotesize adj}}}
\newcommand{\SortedAdjCol}{\mbox{Colors}_{\mbox{\footnotesize split}}}
\newcommand{\SplitCol}{s}
\newcommand{\cdeg}{d^+_{r}}
\newcommand{\ccdeg}{\mbox{cdeg}}
\newcommand{\maxcdeg}{\mbox{max}\ccdeg}

\newcommand{\biggestclass}{b}
\newcommand{\NewCol}{f}		
\newcommand{\maxcol}{k}		
\newcommand{\NumDeg}{\mbox{num}\ccdeg}
\newcommand{\VC}{C}
\newcommand{\NC}{A}
\newcommand{\cdegDom}{D}
\newcommand{\NewColSet}{I}
\newcommand{\COLOR}{\mbox{colour}}
\newcommand{\mincdeg}{\mbox{min}\ccdeg}

\newcommand{\wasinQ}{\mbox{instack}}

\newcommand{\refby}{\preceq}
\newcommand{\strictrefby}{\prec}

\newcommand{\cost}{\operatorname{cost}}
\newcommand{\bcost}{\operatorname{bcost}}

\newcommand{\pictureheight}{3cm}
\newcommand{\gadget}{\mbox{AND}}
\newcommand{\X}{X}
\newcommand{\XX}{\mathcal X}
\newcommand{\Y}{Y}
\newcommand{\YY}{\mathcal Y}

\newcommand{\dist}{\mbox{dist}}
\newcommand{\blk}{\mathcal{B}}
\newcommand{\AND}{\gadget}
\newcommand{\bs}{\setminus}
\newcommand{\disjunion}{\uplus}

\newcommand{\cref}{\simeq}

\definecolor{gruen}{rgb}{0,0.6,0.2}

\newcounter{rbcounter}
\newcommand{\LL}{\textsf{\upshape\mdseries L}}
\newcommand{\LC}{\textsf{\upshape\mdseries C}}

\begin{document}

\title{Tight Lower and Upper Bounds for the Complexity of Canonical Colour Refinement}

\author{
Christoph Berkholz\inst{1} \and 
Paul Bonsma\inst{2}\thanks{Supported by the European Community's Seventh Framework Programme (FP7/2007-2013), grant agreement n$^{\circ}$ 317662.} \and
Martin Grohe\inst{3}\thanks{Supported by the German Research
Foundation DFG Koselleck Grant GR 1492/14-1}
}

\institute{Humboldt-Universit\"at zu Berlin, Germany\\
\email{berkholz@informatik.hu-berlin.de}
\and
University of Twente, the Netherlands\\
\email{p.s.bonsma@ewi.utwente.nl}
\and
RWTH Aachen University, Germany\\
\email{grohe@informatik.rwth-aachen.de}
}

\date{\today}

\maketitle

\begin{abstract}
An assignment of colours to the vertices of a graph is {\em stable} if any two vertices of the same colour have identically coloured neighbourhoods. The goal of {\em colour refinement} is to find a stable colouring that uses a minimum number of colours. This is a widely used subroutine for graph isomorphism testing algorithms, since any automorphism needs to be colour preserving.
We give an $O((m+n)\log n)$ algorithm for finding a {\em canonical} version of such a stable colouring, on graphs with $n$ vertices and $m$ edges. 
We show that no faster algorithm is possible, under some modest assumptions about the type of algorithm, which captures all known colour refinement algorithms.

\medskip
{\em Key words:} Graph isomorphism, colour refinement, partition refinement, canonical labelling
\end{abstract}

\section{Introduction}
{\renewcommand{\thefootnote}{$\dagger$} 
\footnotetext{An extended abstract of this paper has appeared in the proceedings of ESA'13, LNCS 8125, pp 145--156.}}
\emph{Colour refinement} (also known as \emph{naive vertex classification}) is a very simple, yet extremely useful algorithmic routine for graph isomorphism testing. It classifies the vertices by iteratively refining a colouring of the vertices as follows. Initially, all vertices have the same colour. Then in each step of the iteration, two vertices that currently have the same colour get different colours if for some colour $c$ they have a different number of neighbours of colour $c$. The process stops if no further refinement is achieved, resulting in a \emph{stable colouring} of the graph. To use colour refinement as an isomorphism test, we can run it on the disjoint union of two graphs. 
Any isomorphism needs to map vertices to vertices of the same colour. 
So, if the stable colouring differs on the two graphs, that is, if for some colour $c$, the graphs have a different number of vertices of colour $c$, then we know they are nonisomorphic, and we say that colour refinement \emph{distinguishes} the two graphs. Babai, Erd\"os, and Selkow~\cite{baberdsel80} showed that colour 
refinement distinguishes almost all graphs (in the $G(n,1/2)$ model). In fact, they proved the stronger statement that the stable colouring is discrete on almost all graphs, that is, every vertex gets its own colour. On the other hand, colour refinement fails to distinguish any two regular graphs with the same number of vertices, 
such as a $6$-cycle and the disjoint union of two triangles.

Colour refinement is not only useful as a simple isomorphism test in itself, but also as a subroutine for more sophisticated algorithms, both in theory and practice. For example, Babai and Luks's \cite{bab81,babluk83} $O(2^{\sqrt{n\log n}})$-algorithm --- this is still the best known worst-case running time for isomorphism testing --- uses colour refinement as a subroutine, and most practical graph isomorphism tools~(for example, \cite{mck81,darlifsak+04,junkas07,pip11}), starting with McKay's ``Nauty'' \cite{mck81,mck07}, are based on the \emph{individualisation refinement} paradigm (see also~\cite{Miyazaki}). The basic idea of these algorithms is to recursively compute a \emph{canonical labelling} of a given graph, which may already have an initial colouring of its vertices, as follows. We run colour refinement starting from the initial colouring until a stable colouring is reached. 
If the stable  colouring is discrete, then this already gives us a canonical labelling (provided the colours assigned by colour refinement are canonical,
 see 
 below). If not, we pick some colour $c$ with more than one vertex. Then for each vertex $v$ of colour $c$, we modify the stable colouring by assigning a fresh colour to $v$ (that is, we 
``individualise'' $v$) and recursively call the algorithm on the resulting vertex-coloured graph. Then for each $v$ we get a canonically labelled version of our graph, and we return the lexicographically smallest among these. (More precisely, each canonical labelling of a graph yields a canonical string encoding, and we compare these strings lexicographically.) To turn this simple procedure into a practically useful algorithm, various heuristics are applied to prune the search tree. They exploit automorphisms of the graph found during the search. However, crucial for any implementation of such an algorithm is a very efficient colour refinement procedure, because colour refinement is called at every node of the search tree.

Colour refinement can be implemented to run in time $O((n+m)\log n)$, where $n$ is the number of vertices and $m$ the number of edges of the input graph. 
To our knowledge, this was first been proved by Cardon and Crochemore~\cite{carcro82}. Later Paige and Tarjan~\cite[p.982]{paitar87} 
sketched 
a simpler algorithm. Both algorithms are based on the partitioning techniques introduced by Hopcroft~\cite{hop71} for minimising finite automata. However, an issue that is completely neglected in the literature is that, at least for individualisation refinement, we need a version of colour refinement that produces a \emph{canonical colouring}. That is, if $f$ is an isomorphism from a graph $G$ to a graph $H$, then for all vertices $v$ of $G$, $v$ and $f(v)$ should get the same colour in the 
respective 
stable colourings of $G$ and $H$. However, neither of the algorithms analysed in the literature seem to produce canonical colourings. 
We present an implementation of colour refinement that computes a canonical stable colouring in time $O((n+m)\log n)$.
Ignoring the canonical part, our algorithmic techniques are similar to known results: like~\cite{paitar87} and various other papers, we use Hopcroft's strategy of `ignoring the largest new cell', after splitting a cell~\cite{hop71}. Our data structures have some similarities to those described by Junttila and Kaski~\cite{junkas07}. Nevertheless, since~\cite{junkas07} contains no complexity analysis, and~\cite{paitar87} omits various (nontrivial) implementation details, it seems that the current paper gives the first detailed description 
of an $O((m+n)\log n)$ algorithm that uses this 
strategy. On a high level, our algorithm is also quite similar to McKay's {\em canonical} colour refinement algorithm~\cite[Alg.~2.5]{mck81}, but with a few key differences which enable an $O((n+m) \log n)$ implementation.
McKay~\cite{mck81} gave an $O(n^2 \log n)$ implementation using adjacency matrices, which is the previous fastest algorithm for this problem.
Our algorithm is described and analysed in Section~\ref{sec:alg}. In Section~\ref{ssec:extensions}, we discuss extensions: We show how the algorithm can be applied to directed, undirected and edge coloured graphs, and how the complexity bound in fact applies to an entire branch of an individualisation refinement algorithm.

Now the question arises whether colour refinement can be implemented in linear time. After various attempts, we started to believe that it cannot.
Of course with currently known techniques one cannot expect to disprove the existence of a linear time algorithm for the standard (RAM) computation model, or for similar general computation models. 
Instead, we prove a tight lower bound for a restricted, but very broad class of algorithms. In this sense, our result is comparable to the lower bounds for comparison based sorting algorithm.
Actually, our class of \emph{partition-refinement based algorithms} captures all known colour refinement algorithms, and actually every reasonable algorithmic strategy we could think of.
We use the following assumptions. (See Sections~\ref{sec:prelim} and~\ref{sec:lowerbound} for precise definitions.)
Colour refinement algorithms start with a unit partition (which has one {\em cell} $V(G)$), and iteratively refine this until a stable colouring is obtained. This is done using {\em refining operations}: choose a union of current partition cells as {\em refining set} $R$, and choose another (possibly overlapping) union of partition cells $S$. Cells in $S$ are split up if their neighbourhoods in $R$ provide a reason for this. (That is, two vertices in a cell in $S$ remain in the same cell only if they have the same number of neighbours in every cell in $R$.)
This operation requires considering all edges between $R$ and $S$, so the number of such edges is a very reasonable and modest lower bound for the complexity of such a refining step; we call this the {\em cost} of the operation. 
We note that a naive algorithm might choose $R=S=V(G)$ in every iteration. This then requires time $\Omega(mn)$ on graphs that require a linear number of refining operations, such as paths. Therefore, all fast algorithms are based on choosing $R$ and $S$ smartly (and on implementing refining steps efficiently).

For our main lower bound result, we construct a class of instances such that any possible sequence of refining operations that yields the stable partition has total cost at least $\Omega((m+n)\log n)$.
Note that it is surprising that a tight lower bound can be obtained in this model. Indeed, cost {\em upper} bounds in this model would not necessarily yield corresponding algorithms, since firstly we allow the sets $R$ and $S$ to be chosen nondeterministically, and secondly, it is not even clear how to refine $S$ using $R$ in time proportional to the number of edges between these classes. However, as we prove a lower bound, this makes our result only stronger.
An alternative formulation of our lower bound result is to model the class of nondeterministic partition-refinement based algorithms as ``proof system'' and then proves lower bounds on the length of derivations (see the first author's PhD-thesis \cite{ber14} for details).
We formulate the lower bound result for undirected graphs and non-canonical colour refinement, so that it also holds for digraphs, and canonical colour refinement. These results are presented in Section~\ref{sec:lowerbound}.
Our construction also yields corresponding lower bounds for the problems of computing the bisimilarity relation on a transition system and for computing the equivalence classes induced by the 2-variable fragment of first-order logic $\LL^2$ on a structure (see Section~\ref{sec:bisim}).

\section{Preliminaries}
\label{sec:prelim}

For an undirected (simple) graph $G$, $N(v)$ denotes the set of neighbours of $v\in V(G)$, and $d(v)=|N(v)|$ its degree. For a digraph, $N^+(v)$ and $N^-(v)$ denote the out- and in-neighbourhoods, and $d^+(v)=|N^+(v)|$ resp.\ $d^-(v)=|N^-(v)|$ the out- and in-degree, respectively.
A partition $\pi$ of a set $V$ is a set $\{S_1,\ldots,S_k\}$ of pairwise disjoint nonempty subsets of $V$, such that $\cup_{i=1}^k S_i=V$. The sets $S_i$ are called {\em cells} of $\pi$. The {\em order} of $\pi$ is the number of cells $|\pi|$.
A partition $\pi$ is {\em discrete} if every cell has size 1, and {\em unit} if it has exactly one cell.
Given a partition $\pi$ of $V$, and two elements $u,v\in V$, we write $u\approx_{\pi} v$ if and only if there exists a cell $S\in \pi$ with $u,v\in S$. 
We say that a set $V'\subseteq V$ is {\em $\pi$-closed} if it is the union of a number of cells of $\pi$. In other words, if $u\approx_{\pi} v$ and $u\in V'$ then $v\in V'$.
For any subset $V'\subseteq V$, $\pi$ {\em induces} a partition $\pi[V']$ on $V'$, which is defined by $u\approx_{\pi[V']} v$ if and only if $u\approx_{\pi} v$, for all $u,v\in V'$.

Let $G=(V,E)$ be a graph. 
A partition $\pi$ of $V$ is {\em stable} for $G$ if for every pair of vertices $u,v\in V$ with $u\approx_{\pi} v$ and $R\in \pi$, it holds that $|N(u)\cap R|=|N(v)\cap R|$.
If $G$ is a digraph, then $|N^+(u)\cap R|=|N^+(v)\cap R|$ should hold. 
For readability, all further definitions and propositions in this section are formulated for (undirected) graphs, but the corresponding statements also hold for digraphs (replace degrees/neighbourhoods by out-degrees/out-neighbourhoods).
One can see that if $\pi$ is stable and $d(u)\not=d(v)$, 
then $u\not\approx_{\pi} v$, which we will use throughout.

A partition $\rho$ of $V$ {\em refines} a partition $\pi$ of $V$ if for every $u,v\in V$, $u\approx_{\rho} v$ implies $u\approx_{\pi} v$. (In other words: every cell of $\pi$ is $\rho$-closed.)
If $\rho$ refines $\pi$, we write 
$\pi \refby \rho$. 
If in addition $\rho\not=\pi$, then we also write 
$\pi \strictrefby \rho$.
Note that 
$\refby$ is a partial order on all partitions of $V$.
\begin{defi}
Let $G$ be a graph, and let $\pi$ and $\pi'$ be partitions of $V(G)$. For vertex sets $R,S\subseteq V(G)$ that are $\pi$-closed, we say that $\pi'$ is obtained from $\pi$ by a {\em refining operation $(R,S)$} if 
\begin{itemize}
\item 
for every $S'\in \pi$ with $S'\cap S=\emptyset$, it holds that $S'\in \pi'$, and 
\item 
for every $u,v\in S$: $u\approx_{\pi'} v$ if and only if $u\approx_{\pi} v$ and for all $R'\in \pi$ with $R'\subseteq R$, $|N(u)\cap R'|=|N(v)\cap R'|$ holds.
\end{itemize}
\end{defi}

Note that if $\pi'$ is obtained from $\pi$ by a refining operation $(R,S)$, then $\pi\refby \pi'$. We say that the operation $(R,S)$ is {\em effective} if $\pi\strictrefby \pi'$. In this case, at least one cell $C\in \pi$ is {\em split}, which means that $C\not\in \pi'$. Note that an effective refining operation exists for $\pi$ if and only if $\pi$ is unstable. In addition, the next proposition says that if the goal is to obtain a (coarsest) stable partition, then applying any refining operation is safe.

\begin{propo}
\label{propo:RefOpSafe}
Let $\pi'$ be obtained from $\pi$ by a refining operation $(R,S)$. If $\rho$ is a stable partition with $\pi\refby \rho$, then $\pi\refby \pi'\refby \rho$.
\end{propo}
\PF
$\pi\refby \pi'$ follows immediately from the definitions.
Now consider $u,v$ with $u\approx_{\rho} v$, and thus $u\approx_{\pi} v$. Then for any $R'\in \pi$, $d_{R'}(u)=d_{R'}(v)$. This holds because $R'$ is a union of sets in $\rho$, and for all these this property holds since $\rho$ is stable. Therefore, $u\approx_{\pi'} v$.
\QED

A partition $\pi$ is a {\em coarsest} partition for a property $P$ if $\pi$ satisfies $P$, and there is no partition $\rho$ with $\rho \strictrefby \pi$ that also satisfies property $P$.

\begin{propo}
\label{propo:CoarsestStableUnique}
Let $G=(V,E)$ be a graph. For every partition $\pi$ of $V$, there is a unique coarsest stable partition $\rho$ that refines $\pi$.
\end{propo}
\PF
For any partition $\pi$, the discrete partition refines $\pi$ and is stable, so there exists a stable partition that refines $\pi$. Because $\refby$ is a partial order, there exists then at least one coarsest stable partition that refines $\pi$. 
Now suppose there exists a partition $\pi$ for which there exist at least two distinct coarsest stable partitions $\rho_1$ and $\rho_2$ that refine $\pi$. 
Choose such a partition $\pi$ so that $|\pi|$ is maximum.
Clearly, $\pi$ is not stable (otherwise $\rho_1=\pi=\rho_2$). 
So there exists at least one effective refining operation $(R,S)$ that can be applied to $\pi$. For the resulting partition $\pi'$, $|\pi'|>|\pi|$ holds.
By Proposition~\ref{propo:RefOpSafe}, both $\rho_1$ and $\rho_2$ refine $\pi'$ as well. But since $|\pi'|>|\pi|$, this contradicts the choice of $\pi$.
\QED

\section{A Fast Canonical Colour Refinement Algorithm}
\label{sec:alg}

\subsection{Canonical Colouring Methods}

A {\em colouring} of a (di)graph $G$ is a function $\alpha:V(G)\to \mathbb{Z}$. (Note that adjacent vertices may receive the same colour.) 
It is a {\em $k$-colouring} if for every $v\in V(G)$, $\alpha(v)\in \{1,\ldots,k\}$. 
Given a colouring $\alpha$ of $G$ and $i\in \mathbb{Z}$, we denote $C^\alpha_i=\{v\in V(G) \mid \alpha(v)=i\}$. The set $C^\alpha_i$ is called {\em colour class} $i$. If the colouring is clear from the context, we also omit the superscript. 
For any colouring $\alpha$ of a (di)graph $G$, the set $\{C^\alpha_i \mid i\in \mathbb{Z},\ C^\alpha_i\not=\emptyset\}$ is a partition of $V(G)$, which we will denote by $\pi_{\alpha}$.
We will call $\alpha$ {\em unit} or {\em stable} 
if $\pi_{\alpha}$ is unit or stable, respectively. 

Given two (di)graphs $G$ and $G'$, with respective colourings $\alpha$ and $\alpha'$, an isomorphism $h:V(G)\to V(G')$ is {\em colour preserving} for $\alpha$ and $\alpha'$ if for all $v\in V(G)$, $\alpha(v)=\alpha'(h(v))$ holds. 
A {\em colouring method} is a method for obtaining a colouring $\beta$ of a (di)graph $G$, given an initial $k$-colouring $\alpha$. 
(This method can be an algorithm, or simply a definition. Often, the initial colouring $\alpha$ is chosen to be the unit colouring.)
A colouring method (or algorithm) is called {\em canonical} if for any two isomorphic (di)graphs $G$ and $G'$ with initial colourings $\alpha$ resp.\ $\alpha'$ and isomorphism $h:V(G)\to V(G')$, the following holds: if $h$ is colour preserving for $\alpha$ and $\alpha'$, then $h$ is colour preserving for the resulting colourings $\beta$ and $\beta'$.
The resulting colouring $\beta$ itself is also called a {\em canonical colouring of $G$, starting from $\alpha$}. If $\alpha$ is the unit colouring, $\beta$ is simply called a {\em canonical colouring} of $G$. 

For instance, for simple undirected graphs $G$, the degree function $d$, which assigns the colour $d(v)=|N(v)|$ to each $v\in V(G)$, yields canonical colouring of $G$, because every isomorphism maps vertices to vertices of the same degree. (In other words: degrees are {\em isomorphism invariant}.)
Obviously, a canonical colouring method is useful for deducing information about possible isomorphisms between two graphs, especially when the resulting partition $\pi_\beta$ refines the initial partition $\pi_\alpha$. For details on isomorphism testing algorithms based on this idea, we refer to~\cite{mck81,Miyazaki}. 

In this section we give a fast canonical algorithm that for any (di)graph $G$ and colouring $\alpha$ of $G$, yields a colouring $\beta$ of $V(G)$ such that $\pi_\beta$ is the coarsest stable partition that refines $\pi_\alpha$. 
For ease of presentation, we require that the initial colouring $\alpha$ is a {\em surjective} $\ell$-colouring for some value $\ell$ (so every colour in $\{1,\ldots,\ell\}$ occurs at least once). The resulting colouring $\beta$ will then again be a surjective $k$-colouring for some value $k$.
In particular, if we choose $\alpha$ to be the unit colouring, then $\beta$ is a canonical colouring of $G$ such that $\pi_\beta$ is the unique coarsest stable partition of $G$. 
To obtain the most general result, we formulate the algorithm for digraphs. Variants and extensions are discussed in Section~\ref{ssec:extensions}.

\subsection{High-level Description and Correctness Proofs}
\label{sec:correctness}

In Algorithm~\ref{alg:cancolref}, we give a high-level description of our canonical colour refinement algorithm. This is not yet the fast implementation, and in fact, because we do not yet specify which data structures are used to represent the various mathematical objects (sets and functions), no sharp complexity bound can be concluded from it. In the next section, we give a detailed implementation of this algorithm, describe the data structures in detail, and prove the desired complexity bound. Here, we first focus on proving correctness of the algorithm.

In our algorithms, the scope of for loops, while loops and if-then-else statements is indicated by the indentation of blocks; because of space considerations we omit `end for', `end while' and `end if' statements.

\begin{algorithm}
\caption{A canonical colour refinement algorithm}
\label{alg:cancolref}

{\bf Input:} A digraph $G$ on vertex set $\{1,\ldots,n\}$, with surjective 
$\ell$-colouring $\alpha$, and a 
sufficient refining colour set $S\subseteq \{1,\ldots,\ell\}$.

{\bf Output:} A surjective 
canonical $k$-colouring $\beta$ of $G$, starting from $\alpha$, such that $\pi_{\beta}$ is the coarsest stable partition that refines $\pi_{\alpha}$. 

\begin{ntabbing}
\ \=\qquad\=\qquad\=\qquad\=\qquad\=\qquad\=\kill

 \> For $i\in \{1,\ldots,n\}$:\label{la:2}\\
 \> \> $\VC_i:=\{v\in V(G) \mid \alpha(v)=i\}$\label{la:3}\\
 \> $\maxcol:=\ell$\label{la:4}\\
 \> $\refQ:=$ a stack containing all elements of $S$ in increasing order\label{la:5}\\
 \> While $\refQ$ is not empty:\label{la:6}\\
 \> \> $\RefCol:=$pop$(\refQ)$\label{la:7}\\
 \> \> For every vertex $v\in V(G)$:\label{la:8}\\ 
 \> \> \>$\cdeg(v):=|N^+(v)\cap \VC_{\RefCol}|$\label{la:9}\\
 
 \> \> $\SortedAdjCol:=\big\{c\in \{1,\ldots,\maxcol\} \mid \exists v,w\in \VC_c\ \ \cdeg(v)\not=\cdeg(w)\big\}$\label{la:10}\\

 \> \> For all $\SplitCol\in \SortedAdjCol$, in increasing order:\label{la:11}\\
 \> \> \> $\maxcdeg:=\max_{v\in \VC_{\SplitCol}} \cdeg(v)$\label{la:12}\\
 \> \> \> For $i\in \{0,\ldots,\maxcdeg\}$:\label{la:13}\\
 \> \> \> \> $\NumDeg(i):=|\{v\in \VC_{\SplitCol} \mid \cdeg(v)=i\}|$\label{la:14}\\ 
 \> \> \> $\cdegDom:=\{i\in \{0,\ldots,\maxcdeg\} \mid \NumDeg(i)\ge 1\}$\label{la:15}\\
 \> \> \> $\NewColSet:=\{\SplitCol\}\cup \{\maxcol+1,\ldots,\maxcol+|\cdegDom|-1\}$\label{la:16}\\
 \> \> \> Construct a bijection $\NewCol:\cdegDom \rightarrow \NewColSet$ such that $\forall i,j\in D$:\label{la:17}\\
 \> \> \> \> \> $i<j \Rightarrow \NewCol(i)<\NewCol(j)$\\
 \> \> \> For every $v\in \VC_{\SplitCol}$:\label{la:18}\\	
 \> \> \> \> if $\NewCol(\cdeg(v))\not=\SplitCol$ then\label{la:19}\\
 \> \> \> \> \> remove $v$ from $\VC_{\SplitCol}$\label{la:20}\\
 \> \> \> \> \> add $v$ to $\VC_{\NewCol(\cdeg(v))}$\label{la:21}\\
 \> \> \> If $\SplitCol\in \refQ$ then\label{la:22}\\
 \> \> \> \> For $c\in \NewColSet\bs \{\SplitCol\}$, in increasing order:\label{la:23}\\
 \> \> \> \> \> Push$(\refQ,c)$\label{la:24}\\
 \> \> \> else\label{la:25}\\
 \> \> \> \> $\biggestclass:=\min \big\{i\in \{0,\ldots,\maxcdeg\} \mid \forall j\ \ 
	  \NumDeg(i)\ge \NumDeg(j)\big\}$\label{la:26}\\
 \> \> \> \> For $c\in \NewColSet$, in increasing order:\label{la:27}\\
 \> \> \> \> \> If $c\not=\NewCol(\biggestclass)$ then push$(\refQ,c)$\label{la:28}\\
 \> \> \> $\maxcol:=\maxcol+|\NewColSet|-1$\label{la:29}\\
 \> For $i\in \{1,\ldots,\maxcol\}$:\label{la:29b}\\
 \> \> For $v\in \VC_i$:\label{la:29c}\\
 \> \> \> $\beta(v):=i$\label{la:29d}\\
 \> Return the colouring $\beta$ of $G$\label{la:30} 
\end{ntabbing}
\vspace{-0.4cm} 
\end{algorithm}

The input to Algorithm~\ref{alg:cancolref} 
is a digraph $G=(V,E)$, 
with $V=\{1,\ldots,n\}$. 
For every vertex $v\in V$, the sets of out-neighbours $N^+(v)$ and in-neighbours $N^-(v)$ are given. (Alternatively, these can be computed in linear time from the edge list.)
In addition, an $\ell$-colouring $\alpha$ of $G$ and a set $S\subseteq \{1,\ldots,\ell\}$ are given. 
The set $S$ should be a {\em sufficient refining colour set} for $\alpha$, which is a set that satisfies the following property: for any colour class $C^{\alpha}_i$ and two vertices $u,v\in C^{\alpha}_i$, if there exists a colour class $C^{\alpha}_j$ with $|N^+(u)\cap C^{\alpha}_j|\not=|N^+(v)\cap C^{\alpha}_j|$, then there exists a $j'\in S$ such that $|N^+(u)\cap C^{\alpha}_{j'}|\not=|N^+(v)\cap C^{\alpha}_{j'}|$. 
Note that $\{1,\ldots,\ell\}$ trivially forms a sufficient refining colour set for any $\ell$-colouring, but that smarter choices of $S$ may give a faster algorithm (which will be necessary in Section~\ref{ssec:extensions}).

Throughout, the algorithm maintains an (ordered) partition $(\VC_1,\ldots,\VC_{\maxcol})$ of $V(G)$, starting with the partition $(\VC^{\alpha}_1,\ldots,\VC^{\alpha}_{\ell})$ (Lines~\ref{la:2}--\ref{la:4}).
We also view this partition as a colouring, so the sets $\VC_i$ will be called {\em colour classes}, and indices $i\in \{1,\ldots,\maxcol\}$ will be called {\em colours}. 
In the main while-loop (Line~\ref{la:6}), this partition is iteratively refined using refining operations of the form $(R,V)$, where $R=C_r$ for some $r\in \{1,\ldots,\maxcol\}$. We will show that when the algorithm terminates, no effective refining operations are possible on the resulting partition. So the resulting partition is the unique coarsest stable partition of $G$ that refines $\pi_{\alpha}$ (Propositions~\ref{propo:RefOpSafe},~\ref{propo:CoarsestStableUnique}).
The next colour $r$ that is used as {\em refining colour} is chosen using a stack (sequence) $\refQ$ (Line~\ref{la:7}), which contains all colours that still need to be considered. 
For a given refining colour class $C_r$ and any
$v\in V$, call $\cdeg(v):=|N^+(v)\cap C_r|$ the {\em colour degree} of $v$ (with respect to colour $r$). 
Then every colour $s\in \{1,\ldots,\maxcol\}$ will be split up according to colour degrees (in the for-loop of Line~\ref{la:11}). We only consider colours that actually split up, in increasing order. 
When splitting up colour class $C_s$, the new colours will be $s$ and $\maxcol+1,\ldots,\maxcol+d-1$, where $d$ is the number of different colour degrees that occur in $C_s$. These new colours are assigned to the vertices in $C_s$ according to increasing colour degrees (Lines~\ref{la:12}--\ref{la:21}).

It remains to explain how newly introduced colours are added to the stack $\refQ$. 
Initially, $\refQ$ contains all colours in $S$, in increasing order (Line~\ref{la:5}), and whenever new colours are introduced during the splitting of a colour class $C_s$, these are pushed onto the stack $\refQ$, in increasing order (Lines~\ref{la:22}--\ref{la:28}). 
There are however exceptions: for instance, 
if we have already used the vertex set $\VC_{\SplitCol}$ as refining colour class before, and this set is split up into $d$ new colours, then it is not necessary to use all of these new colours as refining colours later; one colour $\biggestclass$ may be omitted from $\refQ$ (Line~\ref{la:28}). 
To obtain a good complexity, we choose $\biggestclass$ such that the size of the corresponding colour class  is maximised, in order to minimise the sizes of the refining colour sets used later during the computation. (This is Hopcroft's trick~\cite{hop71}, which was also used by e.g.~\cite{paitar87}.)

Informally, this algorithm is canonical since at every point, both the (colourings given by the) ordered partition $(\VC_1,\ldots,\VC_k)$ 
and stack $\refQ$ remain canonical; new colours that we assign to vertices, and the order in which colours are considered in the various loops of the algorithm, are completely determined by isomorphism-invariant values such as colour degrees and colour numbers. The order in which vertices of $G$ or neighbour lists are given in the input is irrelevant. A formal proof is given in Lemma~\ref{lem:Canonical} below. We first prove that Algorithm~\ref{alg:cancolref} returns the unique coarsest stable partition, which requires the following invariant.

\begin{propo}
\label{propo:invar}
At the end of every iteration of the for-loop in Line~\ref{la:11} of Algorithm~\ref{alg:cancolref}, $\{\VC_1,\ldots,\VC_{\maxcol}\}$ is a partition of $V(G)$ into nonempty sets, and the set of colours in $\refQ$ is a sufficient refining colour set for the corresponding $\maxcol$-colouring of $G$. 
\end{propo}

\PF
Since new colours correspond to colour degrees that actually occur (Lines~\ref{la:12}--\ref{la:17}), every new colour class will be nonempty. Lines~\ref{la:20} and~\ref{la:21} show that every vertex of $G$ remains part of exactly one colour class. So the algorithm maintains a partition of $V(G)$. 

By definition, the set of colours in $\refQ$ is a sufficient refining colour set before the first iteration. 
We prove that this {\em invariant} is maintained during any iteration of the for-loop, where colour class $\VC_{\SplitCol}$ for $\SplitCol\in\{1,\ldots,\maxcol\}$ is split up (by colour $\RefCol$), into the new colour classes $\VC_{\sigma_1},\ldots,\VC_{\sigma_p}$. Denote $S=\VC_{\SplitCol}$,
as it is at the start of the iteration (so $S=\VC_{\sigma_1}\cup \ldots\cup \VC_{\sigma_p}$).
Because the new colour classes form a partition of the old colour class $S$, for every $z\in V(G)$, it holds that
\begin{equation}
\label{eq:degreesum}
|N^+(z)\cap S|=\sum_{j\in \{1,\ldots,p\}} |N^+(z)\cap \VC_{\sigma_j}|.
\end{equation}

Consider two vertices $u,v\in V(G)$ that are in the same colour class after the refining operation, and therefore also before the refining operation. 
If $|N^+(u)\cap S|\not=|N^+(u)\cap S|$, then there exists an $i\in \{1,\ldots,p\}$ such that $|N^+(u)\cap \VC_{\sigma_i}|\not=|N^+(u)\cap \VC_{\sigma_i}|$ 
(because of~(\ref{eq:degreesum})). 
So if $\SplitCol\in \refQ$, then the invariant is maintained after splitting up the colour, since every new colour is added to $\refQ$ (Lines~\ref{la:23}--\ref{la:24}), and $\SplitCol$ remains in $\refQ$.

So now assume $\SplitCol\not\in \refQ$.
Then every colour in $\{\sigma_1,\ldots,\sigma_p\}$ is added to $\refQ$, except for $i=\NewCol(\biggestclass)$ (Line~\ref{la:28}). 
Then we need to consider the case that $|N^+(u)\cap \VC_{\sigma_i}|\not=|N^+(v)\cap \VC_{\sigma_i}|$, but  $|N^+(u)\cap \VC_{\sigma_j}|=|N^+(v)\cap \VC_{\sigma_j}|$ for all $j\in \{1,\ldots,p\}\bs \{i\}$.
But then also $|N^+(u)\cap S|\not=|N^+(v)\cap S|$ (because of~(\ref{eq:degreesum})). 
Since $\SplitCol\not\in \refQ$, and the invariant held before the refining operation, there exists another colour $j'\in \refQ$ such that $|N^+(u)\cap \VC_{\sigma_{j'}}|\not=|N^+(v)\cap \VC_{\sigma_{j'}}|$. Since this colour remains in $\refQ$, the invariant is also maintained in this case. 
\QED

Using the above proposition, we can prove that Algorithm~\ref{alg:cancolref} computes a coarsest stable colouring, provided that $S$ is a sufficient refining colour set. Recall that this condition is certainly satisfied when choosing $S=\{1,\ldots,\ell\}$. 

\begin{lem}
\label{lem:CoarsestStable} 
Let $G$ be a digraph, $\alpha$ be a surjective $\ell$-colouring of $G$, and let $S\subseteq \{1,\ldots,\ell\}$ be a sufficient refining colour set for $\alpha$. Then Algorithm~\ref{alg:cancolref} computes a surjective $k$-colouring $\beta$ of $G$ such that $\pi_{\beta}$ is the coarsest stable partition that refines $\pi_{\alpha}$. 
\end{lem}

\PF
Let $\omega$ be the coarsest stable partition of $V(G)$ that refines $\pi_{\alpha}$.
The partition $\pi_{\beta}$ given by the algorithm is refined by $\omega$ because it is obtained from $\pi_{\alpha}$ using refining operations (Proposition~\ref{propo:RefOpSafe}). 
The stack $\refQ$ is empty when the algorithm terminates, so the empty set is a sufficient refining colour set at this point (Proposition~\ref{propo:invar}), and therefore $\pi_{\beta}$ is stable. It follows that $\pi_{\beta}$ is equal to $\omega$ (Proposition~\ref{propo:CoarsestStableUnique}). 
At any point, the sets $\VC_i$ for $i\in \{1,\ldots,\maxcol\}$ are nonempty (Proposition~\ref{propo:invar}), so the resulting $k$-colouring $\beta$ is surjective. 
\QED

\begin{lem}
\label{lem:Canonical}
Algorithm~\ref{alg:cancolref} is a canonical colouring algorithm. 
\end{lem}

\PF
Consider two digraphs $G$ and $G'$, with $\ell$-colourings $\alpha$ resp.\ $\alpha'$, and $S\subseteq \{1,\ldots,\ell\}$. 
For $i\in \mathbb{N}$, let $\VC^{G,i}_j$ (resp.\ $\VC^{G',i}_j$) denote the set $\VC_j$ as it is at the start of the $i$-th iteration of the while-loop in Line~\ref{la:6}, when running Algorithm~\ref{alg:cancolref} with input $G,\alpha,S$ (resp.\ $G',\alpha',S$). Let $\refQ^{G,i}$ (resp.\ $\refQ^{G',i}$) denote the stack $\refQ$ as it is at the start of  iteration $i$ of the while-loop in Line~\ref{la:6}, when running Algorithm~\ref{alg:cancolref} with input $G,\alpha,S$ (resp.\ $G',\alpha',S$).

To show that Algorithm~\ref{alg:cancolref} is canonical, we prove by induction over $i$ that
for every isomorphism $h:V(G)\to V(G')$ that is colour-preserving for $\alpha$ and $\alpha'$, the following properties are maintained: 
$\refQ^{G,i}=\refQ^{G',i}$, and for all $c$ and $v\in \VC^{G,i}_c$, it holds that $h(v)\in \VC^{G',i}_c$. 
For $i=1$, the claim follows immediately from how $\refQ$ is initialised (Line~\ref{la:5}), and how the sets $\VC_c$ are initialised (Line~\ref{la:3}). 
We now consider the places in the algorithm where these sets and stacks are modified. In Line~\ref{la:7}, the last element of both $\refQ^{G,i}$ and $\refQ^{G',i}$ is removed, so these sequences stay the same. 
Furthermore, it follows that the same colour is used as refining colour for both $G$ and $G'$ in this iteration. The induction assumption shows that $h$ is a colour preserving isomorphism for the colourings given by the various sets $\VC^{G,i}_c$ and $\VC^{G',i}_c$. So the isomorphism $h$ shows that for every $c$ and every $d$, $\VC^{G,i}_c$ and $\VC^{G',i}_c$ contain the same number of vertices with colour degree $d$. Hence the set $\SortedAdjCol$ is the same for both $G$ and $G'$, and for each colour $c\in \SortedAdjCol$, the values $\maxcdeg$ and $\NumDeg(j)$ (for every $j$) are the same.
Therefore, in every iteration of the for-loop in Line~\ref{la:11}, the sets $\cdegDom$, $\NewColSet$ will be the same for both $G$ and $G'$. The choice of the bijection $f$ in line~\ref{la:17} is unique because of the monotonicity; 
hence $f$ will be the same for $G$ and $G'$ as well.
It follows that when in Lines~\ref{la:20} and~\ref{la:21}, a vertex $v\in V(G)$ is moved from colour class $\VC^{G,i}_{\SplitCol}$ to colour class $\VC^{G,i}_{\NewCol(\cdeg(v))}$, the vertex $h(v)\in V(G')$ is also moved from $\VC^{G',i}_{\SplitCol}$ to $\VC^{G',i}_{\NewCol(\cdeg(v))}$, since $\cdeg(v)=\cdeg(h(v))$. Hence $h$ remains colour preserving for the new partition. 
From the previous observations it also follows that in Line~\ref{la:26}, $\biggestclass$ is chosen to be the same value for both $G$ and $G'$. Therefore, in Lines~\ref{la:28} and~\ref{la:24}, the stack $\refQ$ is modified in the same way for both $G$ and $G'$ (note that in both cases, the colours are added in increasing order). This shows that the claimed properties are maintained in one iteration of the while-loop in Line~\ref{la:6}, so by induction, $h$ is also a colour preserving isomorphism for the final colouring $\beta$ that is returned in Line~\ref{la:30}.
\QED

\subsection{Implementation and Complexity Bound}
\label{ssec:fastimplementation}

We now describe a fast implementation of 
Algorithm~\ref{alg:cancolref}. 
The main idea of the complexity proof is the following: 
one {\em iteration} (of the main while-loop; Line~\ref{la:6} of Algorithm~\ref{alg:cancolref}) consists of popping a refining colour $r$ from the stack $\refQ$, and applying the refining operation $(R,V)$, with $R=C_r$. Below we give implementation details and prove the following lemma:

\begin{lem}
\label{lem:oneiteration}
Algorithm~\ref{alg:cancolref} can be implemented such that one iteration, in which a refining operation $(R,V)$ is applied, takes time
\[
O(|R|+D^-(R)+k\log k),
\]
where $D^-(R)=\sum_{v\in R} d^-(v)$ and $k$ is the number of new colours that are introduced in this iteration. This implementation requires an initialisation step with complexity $O(n)$. 
\end{lem}

Using the above lemma, we can prove the desired complexity bound.
(The main idea is again based on Hopcroft's idea~\cite{hop71}.)

\begin{lem}
\label{lem:complexity}
Algorithm~\ref{alg:cancolref} has an implementation with complexity $O((n+m)\log n)$, where $n=|V(G)|$ and $m=|E(G)|$ for the input digraph $G$. 
\end{lem}

\PF
Consider a vertex $v\in V(G)$. Let $R^v_1,\ldots,R^v_{q}$ denote the refining colour classes $\VC_r$ with $v\in \VC_r$ that are considered throughout the computation, in chronological order. Then we observe that for all $i\in \{1,\ldots,q-1\}$, $|R^v_i|\ge 2|R^v_{i+1}|$ holds. This holds because whenever a set $S=\VC_{\SplitCol}$ is split up into $\VC_{\sigma_1},\ldots,\VC_{\sigma_p}$, where $\SplitCol$ has been considered earlier as a refining colour (so it is not in $\refQ$ anymore), then for all new colours $\sigma_i$ that are added to the stack $\refQ$, $|\VC_{\sigma_i}|\le \frac{1}{2}|S|$ holds, since the largest colour class is not added to $\refQ$. 
Note that if a colour class $\VC_{\sigma_i}$ is subsequently split up before $\sigma_i$ is considered as refining colour, the bound of course also holds.
It follows that every $v\in V(G)$ appears at most $\log_2 n$ times in a refining colour class.
Then we can write
\[
\sum_{R} |R|+D^-(R) \le \sum_{v\in V(G)} (1+d^-(v)) \log_2 n = (n+m) \log_2 n,
\]
where the first summation is over all refining colour classes $R=\VC_r$ considered during the computation.
In addition, the total number of new colours that is introduced is at most $n$, since every colour class, after it is introduced, remains nonempty throughout the computation. So we may write
\[
\sum_{i} k_i \log k_i \le \sum_{i} k_i \log n \le n \log n,
\]
where $k_i$ denotes the number of colours introduced during iteration $i$. 
Combining these bounds with Lemma~\ref{lem:complexity} shows that the total complexity of the algorithm can be bounded by $O(n) + O((n+m) \log n)+O(n\log n) \subseteq O((n+m) \log n)$.
\QED

Combining Lemmas~\ref{lem:CoarsestStable}, \ref{lem:Canonical} and~\ref{lem:complexity} (using $S=\{1,\ldots,\ell\}$), we obtain our main theorem:

\begin{thm}
\label{thm:algmain}
For any digraph $G$ on $n$ vertices and $m$ edges, with surjective $\ell$-colouring $\alpha$, in time $O((n+m)\log n)$ a canonical surjective $k$-colouring $\beta$ of $G$ can be computed such that  $\pi_{\beta}$ is the coarsest stable partition that refines $\pi_{\alpha}$. 
\end{thm}

\subsubsection{Implementation Details}

It remains to prove Lemma~\ref{lem:oneiteration}. 
In Algorithm~\ref{alg:implem} and its subroutine Algorithm~\ref{alg:subrout}, the detailed, fast implementation of Algorithm~\ref{alg:cancolref} is given.
The colour classes $\VC_i$ are represented by doubly linked lists $\VC[i]$, indexed by $i\in \{1,\ldots,n\}$. 
($\VC$ and $\NC$ are arrays containing (pointers to) doubly-linked lists and lists, respectively, indexed by colour numbers $1,\ldots,n$.)
For all lists $L$, we keep track of their length, which we denote by $|L|$.

The first challenge is how to compute the colour degrees $\cdeg(v)$ efficiently for every $v\in V(G)$ (Lines~\ref{la:8} and~\ref{la:9} of Algorithm~\ref{alg:cancolref}), with respect to the refining colour $\RefCol$, and corresponding colour class $R$. For this we use an array $\ccdeg[v]$ of integers, indexed by $v\in \{1,\ldots,n\}$.
We use the following invariant: at the beginning of every iteration, $\ccdeg[v]=0$ for all $v$. Then we can compute these colour degrees by looping over all in-neighbours $w$ of all vertices $v\in R$, 
and increasing $\ccdeg[w]$. At the same time, we compute the maximum colour degree for every colour $c$, using an array $\maxcdeg$ (this is an array of integers indexed by $c\in \{1,\ldots,n\}$), we compute a list $\AdjCol$ of colours $i$ that contain at least one vertex $w\in \VC_i$ with $\ccdeg[w]\ge 1$, and for every such colour $i$, we compute a list $\NC[i]$ of all vertices $w$ with $\ccdeg[w]\ge 1$. 
None of these lists contain duplicates.
See Lines~\ref{li:14}--\ref{li:21} of Algorithm~\ref{alg:implem}. 
This implementation is correct because we also maintain the following invariant:
at the beginning of every iteration, $\maxcdeg[c]=0$ and $\NC[c]$ is an empty list, for every $c$, $\AdjCol$ is an empty list, and flags are maintained for colours to keep track of membership in $\AdjCol$. 
To maintain this invariant, we reset all of these data structures again at the end of every iteration (Lines~\ref{li:35}--\ref{li:40}). Note that it suffices to only reset $\ccdeg[v]$ for vertices $v$ that occur in some list $\NC[c]$ (Lines~\ref{li:36}--\ref{li:37}).

Next, we address how we can consider all colours that split up in one iteration, in canonical (increasing) order (see Lines~\ref{la:10},\ref{la:11} of Algorithm~\ref{alg:cancolref} and Lines~\ref{li:22}--\ref{li:33}).
To this end, we compute a new list $\SortedAdjCol$, which represents the subset of $\AdjCol$ containing
all colours that actually split up. 
This is necessary since this list needs to be sorted, in order to consider the colours in canonical order (in the for-loop in Line~\ref{li:33}).
By ensuring that all colours in $\SortedAdjCol$ split up, we have that $|\SortedAdjCol|\le k$ (where $k$ is the number of colours introduced in this iteration), and therefore we can afford to sort this list. This can be done using any list sorting algorithm of complexity $O(k \log k)$, such as merge sort.  
To compute which colours split up, we compute for every colour in $c\in \AdjCol$ the maximum colour degree $\maxcdeg[c]$ and minimum colour degree $\mincdeg[c]$. The values $\maxcdeg[c]$ were computed before. Observe that we have $\mincdeg[c]=0$ if $|\NC[c]|<|\VC[c]|$. Otherwise, we can afford to compute $\mincdeg[c]$ by iterating over $\NC[c]$
(see Lines~\ref{li:22}--\ref{li:28}).

\begin{algorithm}
\caption{A fast implementation of Algorithm~\ref{alg:cancolref}}
\label{alg:implem}
\begin{ntabbing}
\ \=\qquad\=\qquad\=\qquad\=\qquad\=\kill

\> For $c\in \{1,\ldots,n\}$:\label{li:1}\\
\> \> $\VC[c]:=$ an empty doubly linked list\label{li:2}\\
\> \> $\NC[c]:=$ an empty list\label{li:3}\\
\> \> $\maxcdeg[c]:=0$\label{li:4}\\
\> For $v\in \{1,\ldots,n\}$:\label{li:5} \algcomm{$V(G)=\{1,\ldots,n\}$}\\
\> \> append $v$ to $\VC[\alpha(v)]$\label{li:6}\\
\> \> $\ccdeg[v]:=0$\label{li:7}\\
\> \> $\COLOR[v]:=\alpha(v)$\label{li:8}\\
\> $\maxcol:=\ell$\label{li:9}\\
\> $\refQ:=$ a stack containing all elements of $S$\label{li:10}\\
\> Sort $\refQ$\label{li:10b}\\
\> $\AdjCol:=$ an empty doubly linked list\label{li:11}\\
\> While $\refQ$ is not empty:\label{li:12}\\
\> \> $\RefCol:=$pop$(\refQ)$\label{li:13}\\
\> \> For $v\in \VC[\RefCol]$:\label{li:14}\\
\> \> \> For $w\in N^-(v)$:\label{li:15}\\
\> \> \> \> $\ccdeg[w]:=\ccdeg[w]+1$\label{li:16}\\
\> \> \> \> If $\ccdeg[w]=1$ then append $w$ to $\NC[\COLOR[w]]$\label{li:17}\\
\> \> \> \> If $\COLOR[w]\not\in \AdjCol$ then\label{li:18} \algcomm{Maintain flag for this test}\\
\> \> \> \> \> append $\COLOR[w]$ to $\AdjCol$\label{li:19}\\
\> \> \> \> If $\ccdeg[w]>\maxcdeg[\COLOR[w]]$ then\label{li:20}\\
\> \> \> \> \> $\maxcdeg[\COLOR[w]]:=\ccdeg[w]$\label{li:21}\\
\> \> For $c\in \AdjCol$:\label{li:22}\\
\> \> \> If $|\VC[c]|\not=|\NC[c]|$ then\label{li:23} \algcomm{Maintain list lengths}\\
\> \> \> \> $\mincdeg[c]:=0$\label{li:24}\\
\> \> \> else\label{li:25}\\
\> \> \> \> $\mincdeg[c]:=\maxcdeg[c]$\label{li:26}\\
\> \> \> \> For $v\in \NC[c]$:\label{li:27}\\
\> \> \> \> \> if $\ccdeg[v]<\mincdeg[c]$ then $\mincdeg[c]:=\ccdeg[v]$\label{li:28}\\
\> \> $\SortedAdjCol:=$ an empty list\label{li:29}\\
\> \> For $c\in \AdjCol$:\label{li:30}\\
\> \> \> If $\mincdeg[c]<\maxcdeg[c]$ then append $c$ to $\SortedAdjCol$\label{li:31}\\
\> \> Sort $\SortedAdjCol$\label{li:32}\\
\> \> For all $\SplitCol\in \SortedAdjCol$, in increasing order:\label{li:33}\\
\> \> \> {\sc SplitUpColour($\SplitCol$)}\label{li:34} \algcomm{Subroutine, see next page}\\
\> \> \algincomm{Reset the attributes here, ready for the next iteration (next choice of $\RefCol$)}:\\
\> \> For $c\in \AdjCol$:\label{li:35}\\
\> \> \> For $v\in \NC[c]$:\label{li:36}\\
\> \> \> \> $\ccdeg[v]:=0$\label{li:37}\\
\> \> \> $\maxcdeg[c]:=0$\label{li:38}\\
\> \> \> $\NC[c]:=$ an empty list\label{li:39}\\
\> \> \> Remove $c$ from $\AdjCol$\label{li:40}\\
\> \algincomm{End of algorithm:}\\
\> Return the array $\COLOR$.\label{li:41} \algcomm{This is the final colouring $\beta$}\\

\end{ntabbing}
\vspace{-0.4cm} 
\end{algorithm}

Finally, we need to show how a single colour class $S=\VC[\SplitCol]$
can be split up efficiently, and how the appropriate new colours can be added to the stack $\refQ$ in the proper order (Lines~\ref{la:11}--\ref{la:29} of Algorithm~\ref{alg:cancolref}).
The details of this procedure are given in Algorithm~\ref{alg:subrout}.
Firstly, for every relevant $d$, we compute how many vertices in $\VC[\SplitCol]$ have colour degree $d$. These values are stored in an array $\NumDeg[d]$, indexed by $d\in \{0,\ldots,\maxcdeg[\SplitCol]\}$ (Lines~\ref{ls:2}--\ref{ls:6}).
Using this array $\NumDeg$, we can easily compute the (minimum) colour degree $\biggestclass$ that occurs most often in $S$ (Lines~\ref{ls:7}--\ref{ls:9}), 
which corresponds to the new colour that is possibly not added to $\refQ$.
Using $\NumDeg$, we can also easily construct an array $\NewCol$, indexed by $d\in \{0,\ldots,\maxcdeg[\SplitCol]\}$, which represents the mapping from colour degrees that occur in $S$ to newly introduced colours, or to the current colour $\SplitCol$
(Lines~\ref{ls:11}--\ref{ls:19}). 
Finally, we can 
move all vertices $v\in\NC[\SplitCol]$ from $\VC[\SplitCol]$ to $\VC[i]$, where $i=\NewCol[\ccdeg[v]]$ is the new colour that corresponds to the colour degree of $v$
(Lines~\ref{ls:20}--\ref{ls:24}). 
Note that looping over $\NC[\SplitCol]$ suffices, because if there are vertices in $\VC[\SplitCol]$ with colour degree 0, then these keep the same colour, and thus do not need to be addressed. This fact is essential since the number of such vertices may be too large to consider, for our desired complexity bound.
In conclusion, Algorithms~\ref{alg:implem} and~\ref{alg:subrout} are indeed implementations of Algorithm~\ref{alg:cancolref}. We now prove Lemma~\ref{lem:oneiteration} by analysing the complexity. 

\begin{algorithm}
\caption{Subroutine {\sc SplitUpColour$(\SplitCol)$}}
\label{alg:subrout}
\begin{ntabbing}
\ \=\qquad\=\qquad\=\qquad\=\qquad\=\qquad\=\qquad\=\kill
\> $\maxcdeg:=\maxcdeg[\SplitCol]$\label{ls:1}\\
\> For $i\in [1,\ldots,\maxcdeg]$:\label{ls:2}\\
\> \> $\NumDeg[i]:=0$\label{ls:3}\\
\> $\NumDeg[0]:=|\VC[\SplitCol]|-|\NC[\SplitCol]|$\label{ls:4}\\
\> For $v\in \NC[\SplitCol]$:\label{ls:5}\\
\> \> $\NumDeg[\ccdeg[v]]:=\NumDeg[\ccdeg[v]]+1$\label{ls:6}\\
\> $\biggestclass:=0$\label{ls:7}\\
\> For $i\in [1,\ldots,\maxcdeg]$:\label{ls:8}\\
\> \> If $\NumDeg[i]>\NumDeg[\biggestclass]$ then $\biggestclass:=i$\label{ls:9}\\

\> If $\SplitCol\in \refQ$ then $\wasinQ:=1$ else $\wasinQ:=0$\label{ls:10} \algcomm{maintain flag for this test}\\

\> For $i\in [0,\ldots,\maxcdeg]$:\label{ls:11}\\
\> \> If $\NumDeg[i]\ge 1$ then\label{ls:12}\\
\> \> \> If $i=\mincdeg[\SplitCol]$ then\label{ls:13}\\
\> \> \> \> $\NewCol[i]:=\SplitCol$\label{ls:14}\\
\> \> \> \> If $\wasinQ=0$ and $b\not=i$ then push$(\refQ,\NewCol[i])$\label{ls:15}\\
\> \> \> else\label{ls:16}\\
\> \> \> \> $\maxcol:=\maxcol+1$\label{ls:17}\\
\> \> \> \> $\NewCol[i]:=\maxcol$\label{ls:18}\\
\> \> \> \> If $\wasinQ=1$ or $i\not=b$ then push$(\refQ,\NewCol[i])$\label{ls:19}\\

\> For $v\in \NC[\SplitCol]$:\label{ls:20}\\
\> \> If $\NewCol[\ccdeg[v]]\not=\SplitCol$ then\label{ls:21}\\
\> \> \> Delete $v$ from $\VC[\SplitCol]$\label{ls:22}\\
\> \> \> Append $v$ to $\VC[\NewCol[\ccdeg[v]]]$\label{ls:23}\\
\> \> \> $\COLOR[v]:=\NewCol[\ccdeg[v]]$\label{ls:24}
\end{ntabbing}
\vspace{-0.4cm} 
\end{algorithm}

\noindent
{\bf Proof of Lemma~\ref{lem:oneiteration}:}
The given implementation uses a number of arrays of length $n$, either containing integers ($\ccdeg$, $\COLOR$, $\maxcdeg$, $\NumDeg$, $\NewCol$), or containing (pointers to) lists/doubly linked lists ($\VC$, $\NC$). 
All of these arrays can be initialised in time $O(n)$. In general, the initialisation steps (Lines~\ref{li:1}--\ref{li:11}) take time $O(n)$ (for Line~\ref{li:10b}, use bucket sort).

We first consider the complexity of the subroutine {\sc SplitUpColour$(\SplitCol)$}, given in Algorithm~\ref{alg:subrout}. 
We prove that it terminates in time $O(D^+_R(S))$, where $R=\VC[r]$ denotes the refining colour class, $S=\VC[\SplitCol]$ denotes the class to be split up, and $D^+_R(S)=\sum_{v\in S} |N^+(v)\cap R|$. 
Every (non-loop) line takes constant time. For the list deletion (Line~\ref{ls:22}), this requires a proper implementation of doubly linked lists. The test in Line~\ref{ls:10} whether $\SplitCol\in \refQ$ can be done in constant time by maintaining a 0/1 flag for every colour, which indicates whether the colour is in $\refQ$. Since colours are added to and deleted from the stack $\refQ$ one by one, maintaining these flags is no problem.
All for-loops in Algorithm~\ref{alg:subrout} are repeated either $\maxcdeg[\SplitCol]$ times or $|\NC[\SplitCol]|$ times. Both values are bounded by $D^+_R(S)$.
So the total complexity of one call to the subroutine can be bounded by $O(D^+_R(S))$.

Now consider the complexity of one while loop iteration of Algorithm~\ref{alg:implem}.
The first two (nested) for loops (Lines~\ref{li:14}--\ref{li:21}) take time $O(|R|+D^-(R))$. 
This holds because in total, $D^-(R)$ choices of $w$ are considered, and the operations for every such choice take constant time. The test in Line~\ref{li:18} can be implemented in constant time using a 0/1 flag that keeps track of whether a colour appears in $\AdjCol$. Since elements are added to and deleted from $\AdjCol$ one by one (Lines~\ref{li:19},~\ref{li:40}), maintaining these flags is again no problem. 

Since $|\AdjCol|\le D^-(R)$, the complexity of the for loops in Lines~\ref{li:22} and~\ref{li:30} can be bounded by $O(D^-(R))$.
Sorting $\SortedAdjCol$ takes time $O(k \log k)$, when using e.g. merge sort, since $|\SortedAdjCol|\le k$ (every colour in $\SortedAdjCol$ will split up and thus introduce at least one new colour).
One call to the subroutine {\sc SplitUpColour$(\SplitCol)$} takes time $O(D^+_R(S))$, with $S=\VC[\SplitCol]$, as shown above.
Since
\[
\sum_{\SplitCol\in \SortedAdjCol} D^+_R(\VC[\SplitCol]) \le D^-(R),
\]
the complexity of the for-loop in Line~\ref{li:33} can be bounded by $O(D^-(R))$.
The complexity of the last for-loop (Line~\ref{li:35}) can also be bounded by $O(D^-(R))$. Note in particular that in total, at most $D^-(R)$ choices of $v$ are considered in Line~\ref{li:37}.
This shows that the complexity of one iteration of the while-loop can be bounded by $O(|R|+D^-(R)+k\log k)$.
\QED

\subsection{Extensions, Generalisations and Variants}
\label{ssec:extensions}

\subsubsection{Stack vs.\ queue} 
In our algorithm, we use a stack to select the next colour that should be used for the next refining operation, whereas previous similar algorithms use a queue~\cite{mck81,paitar87}. 
Firstly, we remark that if we replace the stack by a queue, it can easily be checked that all of the claims proved in the previous sections still hold. So the best choice is determined by other concerns, which we now shortly discuss.

Using a queue gives the nice property that during the algorithm execution, all of the following `standard' partitions will be generated: given an initial partition $\pi=\pi_0$ of the vertices $V$ of a graph $G$, for every $i\ge 0$ one can define $\pi_{i+1}$ to be the partition obtained from $\pi_i$ using the refining operation $(V,V)$. The coarsest stable partition of $G$ that refines $\pi$ is now the first partition $\pi_i$ with $\pi_i=\pi_{i+1}$. This characterisation is sometimes used as an alternative definition of coarsest stable partitions. One can verify that when using a queue, for every $i$ a colouring $\alpha$ with $\pi_{\alpha}=\pi_i$ will be generated during the execution of the algorithm. 

When using a stack, the behaviour of the algorithm seems somewhat less predictable. Nevertheless, this yields a `depth-first' type of strategy that tends to give very small colour classes much quicker, which seems an advantage. In our own (limited) computational studies, we observed that using a stack was never worse than using a queue, and in some cases significantly better. Furthermore, we had an earlier lower bound example construction that required time $\Omega((n+m)\log n)$ for a queue-based algorithm, but could be solved in time $O(n+m)$ using a stack-based algorithm. For these reasons, we would recommend using a stack.

\begin{algorithm}
\caption{Iterative Colour Refinement}
\label{alg:fullbranch}
{\bf Input:} A digraph $G$ with $V(G)=\{1,\ldots,n\}$.
\begin{ntabbing}
\ \=\qquad\=\qquad\=\kill
\> Compute a surjective canonical coarsest stable $k$-colouring $\beta$ of $G$.\label{lb:1}\\
\> While $k<n$:\label{lb:2}\\
\> \> $\alpha:=\beta$\label{lb:3}\\
\> \> $\ell:=k$\label{lb:4}\\
\> \> Choose a vertex $v$ with a non-unique colour in $\alpha$.\label{lb:5}\\
\> \> $\ell:=\ell+1$\label{lb:6new}\\
\> \> $\alpha(v):=\ell$\label{lb:7new}\\
\> \> Compute a surjective canonical $k$-colouring $\beta$ of $G$ such that\label{lb:8}\\
\> \> \> $\pi_{\beta}$ is the coarsest stable partition that refines $\pi_{\alpha}$.\\
\> Return $\beta$\label{lb:9}\\
\end{ntabbing}
\vspace{-0.4cm} 
\end{algorithm}

\subsubsection{The Complexity of Iterative Refinement}
Consider Algorithm~\ref{alg:fullbranch}. This algorithm takes as input a digraph $G$ on $n$ vertices, and returns a discrete colouring $\beta$ of $G$, or more precisely: a surjective $n$-colouring of $G$. 
This colouring is not canonical, since in Line~\ref{lb:5}, an arbitrary vertex is chosen to be {\em individualised}, that is, to receive a unique colour. So by itself this algorithm is not very interesting (there are easier ways to obtain an arbitrary discrete colouring of $G$). 
However, it corresponds to {\em one recursion branch} of various state of the art canonical labelling algorithms, based on the algorithm introduced by McKay~\cite{mck81}. 
We now shortly sketch how one should modify this algorithm (into a recursive algorithm) to obtain such a canonical labelling algorithm:
In Line~\ref{lb:5}, instead choose a colour class $C^\alpha_i$ of the current colouring $\alpha$ with $|C^{\alpha}_i|\ge 2$. We {\em branch} on this colour class, as follows: for {\em every} $v\in C^{\alpha}_i$, continue with a separate branch of the algorithm where $v$ is individualised (Line~\ref{lb:7new}), and a new stable colouring is computed (as shown in Line~\ref{lb:8}). Continuing recursively this way, one obtains a number of discrete colourings of $G$; one for every leaf of the recursion tree. 
A canonical discrete colouring of $G$ can be obtained by choosing one of these colourings that maximises some value. For instance, consider the adjacency matrix representation of $G$ where rows and columns are ordered according to the colour numbers, and view this as a binary number in the straightforward way. 
This is the basic algorithm; by keeping track of automorphisms of the graph, there are various ways to speed up the algorithm by pruning the recursion tree. For more details, we refer to~\cite{mck81,darlifsak+04,junkas07,pip11,mck07}. 

The algorithm for obtaining a canonical discrete colouring $\beta$ for a digraph $G$ sketched above does not terminate in polynomial time for all graphs $G$. (If it did, this would yield a polynomial time isomorphism test: for two digraphs $G$ and $G'$, compute canonical discrete $n$-colourings $\beta$ and $\beta'$, respectively. Since $\beta$ and $\beta'$ are discrete $n$-colourings, they define a unique colour preserving bijection $h:V(G)\to V(G')$. Since  $\beta$ and $\beta'$ are canonical, $G$ and $G'$ are isomorphic if and only if $h$ is an isomorphism.) Examples are known where such an algorithm will consider an exponential number of branches~\cite{Miyazaki}. 
Nevertheless, a single branch of this algorithm (as shown in Algorithm~\ref{alg:fullbranch}) terminates quickly. From~\cite{mck81} it follows that Algorithm~\ref{alg:fullbranch} has an implementation that terminates in time $O(n^2 \log n)$.
Using our results, we can show that it has an $O((n+m) \log n)$ implementation.

\begin{thm}
Algorithm~\ref{alg:fullbranch} can be implemented such that it terminates in time $O((n+m) \log n)$, where $n=|V(G)|$ and $m=|E(G)|$ for the input graph $G$.
\end{thm}

\PF
The main part of the computation occurs in Lines~\ref{lb:1} and~\ref{lb:8}, where we compute a surjective canonical $k$-colouring $\beta$ such that $\pi_{\beta}$ refines $\pi_{\alpha}$, for a given surjective $\ell$-colouring $\alpha$ (in Line~\ref{lb:1}, the unit colouring is chosen for $\alpha$). 
For this we use the fast implementation of Algorithm~\ref{alg:cancolref}, given in Section~\ref{ssec:fastimplementation}. To obtain the desired complexity, we make the following simple changes, compared to Algorithm~\ref{alg:cancolref}: we do not initialise the sets $\VC_i$ and stack $\refQ$ every time we call the algorithm (Lines~\ref{la:2}--\ref{la:5}), and do not explicitly compute the new colouring $\beta$ (Lines~\ref{la:29b}--\ref{la:30}). In addition, we do not actually copy the the colouring $\beta$ (Line~\ref{lb:3} of Algorithm~\ref{alg:fullbranch}). 
Instead, we initialise these sets once, keep working with the same sets $\{\VC_1,\ldots,\VC_{\maxcol}\}$ throughout different iterations of the while-loop in Algorithm~\ref{alg:fullbranch}, and and only compute the corresponding colouring $\beta$ at the very end of the algorithm. 
Whenever we individualise a vertex $v$ by assigning it a new colour (Line~\ref{lb:7new} of Algorithm~\ref{alg:fullbranch}), we move $v$ from its previous colour class $\VC_i$ to the new colour class $\VC_{\ell+1}$. In addition, we update the stack $\refQ$, which is currently empty, to contain the single colour $\ell+1$. (This can both be done in constant time.) 

We now argue that for computing the next stable colouring $\beta$ (Line~\ref{lb:8}), it is sufficient that $\refQ$ contains only the colour $\ell+1$. 
Denote by $\alpha_1$ the stable $\ell$-colouring before this step (with $\alpha_1(v)=i$), and by $\alpha_2$ the new $(\ell+1)$-colouring (with $\alpha_2(v)=\ell+1$). 
Consider the colour classes $C^{\alpha_1}_i$, $C^{\alpha_2}_i$ and $C^{\alpha_2}_{\ell+1}$, so $\{C^{\alpha_2}_i,C^{\alpha_2}_{\ell+1}\}$ is a partition of $C^{\alpha_1}_i$. 
Consider any two vertices $u,v$ with $\alpha_2(u)=\alpha_2(v)$. If $|N^+(u)\cap C^{\alpha_2}_i|\not=|N^+(v)\cap C^{\alpha_2}_i|$, then $|N^+(u)\cap C^{\alpha_2}_{\ell+1}|\not=|N^+(v)\cap C^{\alpha_2}_{\ell+1}|$, since $|N^+(u)\cap C^{\alpha_1}_i|=|N^+(v)\cap C^{\alpha_1}_i|$ (because $\alpha_1$ is stable). 
We conclude that $\{\ell+1\}$ is a sufficient refining colour set for $\alpha_2$, so Algorithm~\ref{alg:cancolref} will compute the desired stable colouring $\beta$ when $\refQ$ is initialised like this (Lemma~\ref{lem:CoarsestStable}). 

We can now use the same argument as in the proof of Theorem~\ref{thm:algmain} to show that the total complexity of all calls to Algorithm~\ref{alg:cancolref} (without the initialisation steps, as described above) is bounded by $O((n+m)\log n)$. Indeed, for every vertex $v\in V(G)$, if $R^v_1,\ldots,R^v_{q}$ denote the refining colour classes $\VC_r$ with $v\in \VC_r$ that are considered throughout the entire computation, in chronological order, then again for all $i\in \{1,\ldots,q-1\}$ it holds that $|R^v_i|\ge 2|R^v_{i+1}|$. If $v$ is the vertex that is individualised in Line~\ref{lb:7new} (of Algorithm~\ref{alg:fullbranch}), then this holds because the next refining colour class that contains $v$ has size one, whereas the previous colour class that contained $v$ had size at least two (because $v$ was chosen with a non-unique colour in Line~\ref{lb:5}). 
In all other cases, the argument given in the proof of Theorem~\ref{thm:algmain} applies. Following that proof, this shows that the total complexity of all refining operations done in Algorithm~\ref{alg:fullbranch} can be bounded by $O((n+m)\log n)$.  

It remains to bound the complexity of the other steps of Algorithm~\ref{alg:fullbranch}. As described above, the various sets are initialised only once, and the final colouring $\beta$ is computed only once, so this only adds a term $O(n)$ to the complexity. In addition, all steps in the while-loop (Line~\ref{lb:2}) other than the stable colouring computation in Line~\ref{lb:8} can be done in constant time, since we do not actually copy the colouring (Line~\ref{lb:3}). 
For the selection of the vertex $v$ in Line~\ref{lb:5}, this claim is not entirely obvious, but one can observe that during the computation, one can maintain a doubly linked list that contains the colours of all colour classes of size at least two. This list can be updated in constant time whenever vertices are recoloured (so it does not change the total asymptotic complexity), and it can be used to select a vertex in Line~\ref{lb:5} in constant time. The while-loop in Line~\ref{lb:2} terminates after at most $n$ iterations. In total, this shows that Algorithm~\ref{alg:fullbranch} has an implementation with complexity $O(n)+O(n)+O((n+m)\log n) \subseteq O((n+m)\log n)$. 
\QED

We remark that in practice, one might wish to use smarter methods to select the vertex $v$ to be individualised (Line~\ref{lb:5}), or more generally, to select the nontrivial colour class on which the recursive canonical labelling algorithm should branch. For instance, one can always branch on the smallest nontrivial colour class, or on the largest colour class\footnote{Our own computational tests showed, somewhat surprisingly, that branching on the {\em largest} colour class is clearly the best strategy of these two.}. 
In that case, an efficient heap-based priority queue implementation (see e.g.~\cite{fretar87}) can be used instead of a doubly-linked list to keep track of the sizes of colour classes, to attain the above complexity.

\subsubsection{Alternative Stability Criteria}
We formulated our results only for digraphs, with stability defined only in terms of out-neighbours. 
We now summarise how our results should be modified to accommodate alternative stability criteria.
\begin{thm}
\label{thm:algmain_undir}
For any undirected graph $G$ on $n$ vertices with $m$ edges, in time $O((n+m)\log n)$ a canonical coarsest stable colouring can be computed.
\end{thm}
\PF
For an undirected graph $G$, denote by $G^*$ the digraph with $V(G^*)=V(G)$, constructed by replacing every undirected edge by two directed edges in both directions. 
Observe that a colouring $\alpha$ is stable for $G$ if and only if it is stable for $G^*$, so we can use the fast implementation of Algorithm~\ref{alg:cancolref} on input $G^*$ to compute a coarsest stable colouring of $G$. 
Next, observe that a bijection $h:V(G)\to V(H)$ is an isomorphism from $G$ to $H$ if and only if it is a (digraph) isomorphism from $G^*$ to $H^*$. It follows that the computed colouring is a {\em canonical} coarsest stable colouring.
\QED

For a positive integer $p$, we define a {\em $p$-edge coloured digraph} $G$ to be a tuple 
$(V,E,c)$ where $(V,E)$ is a digraph that may have parallel edges and/or loops, and $c:E\to \{1,\ldots,p\}$ is an edge colouring of $G$. For $e\in E$, we write $e=(u,v)$ to denote that $e$ is an edge from $u$ to $v$.  
For $j\in \{1,\ldots,p\}$, $v\in V$ and $C\subseteq V$, denote $d^+_j(v,C)=|\{ e\in E \mid \exists w\in C \ e=(v,w) \wedge c(e)=j\}|$ (the number of edges of colour $j$, leaving $v$, with head in $C$).
A (vertex) $\ell$-colouring $\alpha$ of $G$ is called {\em edge-colour stable} if for all $u,v\in V$ with $\alpha(u)=\alpha(v)$, all $j\in \{1,\ldots,p\}$ and all $i\in \{1,\ldots,\ell\}$, it holds that 
$d^+_j(u,C^{\alpha}_i) = d^+_j(v,C^{\alpha}_i)$. 
For two $p$-edge coloured digraphs $G=(V,E,c)$ and $G'=(V',E',c')$, a bijection $h:V\to V'$ is called an {\em isomorphism} if for all $j\in \{1,\ldots,p\}$ and $u,v\in V$ (possibly the same), it holds that 
the number of edges of colour $j$ from $u$ to $v$ equals the number of edges of colour $j$ from $h(u)$ to $h(v)$.
Using this notion of isomorphism, canonical colouring methods/canonical colourings for edge coloured digraphs are defined the same as before.
\begin{thm}
\label{thm:algmain_edgecol}
Let $G=(V,E,c)$ be an edge coloured digraph with $n=|V|$ and $m=|E|$. 
In time $O((n+m)\log (n+m))$,
a canonical coarsest edge-colour stable colouring can be computed for $G$. 
\end{thm}
\PF
In time $O(n+m)$, we can construct the following digraph $G'$ from $G$, with vertex colouring $\alpha$:
Start with the vertex set $V$ (the {\em original vertices}), and for every edge $e\in E$ with $e=(u,v)$, add a vertex $v_e$ (the {\em new vertices}) and two edges $(u,v_e)$ and $(v_e,v)$. Assign colour $\alpha(v_e)=c(e)$ to the new vertices, and colour $\alpha(v)=0$ to the original vertices $v\in V$. 

We will now show that a colouring $\beta$ of $G$ is edge-colour stable if and only if there exists a stable colouring $\beta'$ for $G'$ that refines $\alpha$, that coincides with $\beta$ on $V$. 

Consider an edge-colour stable colouring $\beta$ of $G$. We extend it to a colouring $\beta'$ of $G'$, as follows: for each new vertex $v_e$ that corresponds to an edge $e=(u,v)$, assign the tuple $(c(e),\beta(v))$. Extend $\beta$ by assigning new colours to the new vertices, according to the lexicographical order of these tuples. (So two new vertices receive the same colour if and only if they are assigned the same tuple, and a new vertex and an original vertex never receive the same colour.) The resulting colouring $\beta'$ of $G'$ clearly refines $\alpha$, and is stable for $G'$: for every vertex colour $i$ used by $\beta$, vertex $u\in V$ and edge colour $j\in \{1,\ldots,p\}$, the number $d^+_j(u,C^{\beta}_i)$ (with respect to $G$) equals the number of out-neighbours of $u$ in $G'$ that have the colour corresponding to the tuple $(j,i)$. For the new vertices of $G'$, the stability criterion follows easily. 

For the other direction, consider a stable colouring $\beta'$ of $G'$ that refines $\alpha$, and define $\beta$ to be the restriction of $\beta'$ to $V$. We argue that $\beta$ is edge-colour stable for $G$. For two new vertices $v_e$ and $v_f$ of $G'$, with respective out-neighbours $x$ and $y$, we have that $\beta'(v_e)=\beta'(v_f)$ implies  $\beta'(x)=\beta'(y)$ and $\alpha(v_e)=\alpha(v_f)$, so $c(e)=c(f)$. This can be used to conclude that for any two vertices $u,v\in V$, colour $i$ and edge colour $j$, if $\beta(u)=\beta(v)$ then $d^+_j(u,C^{\beta}_i)=d^+_j(v,C^{\beta}_i)$. So $\beta$ is edge-colour stable for $G$.

It follows that a coarsest edge-colour stable colouring $\beta$ of $G$ corresponds to a coarsest stable colouring $\beta'$ of $G'$ that refines $\alpha$. Since we can compute such a colouring $\beta'$ in a canonical way, we can compute such a colouring $\beta$ in a canonical way (Theorem~\ref{thm:algmain}). It remains to consider the complexity. The graph $G'$ and colouring $\alpha$ can be constructed from $G$ in time $O(n+m)$. It has $n+m$ vertices, and $2m$ edges. So $\beta'$ can be computed in time $O((n+3m) \log (n+m))=O((n+m) \log (n+m))$ time, by Theorem~\ref{thm:algmain}.\QED

We remark that for any class of edge-coloured digraphs where the number of edges is polynomially bounded in the number vertices (so they satisfy $m\in O(n^d)$ for a constant $d$), we can write  $\log (n+m)\in O(\log n^d)=O(\log n)$. So for such a graph class, the above lemma shows that a canonical coarsest edge-colour stable colouring can again be computed in time $O((n+m)\log n)$.

The above theorem can be used for various stronger isomorphism tests. 
We now give details for one of these. 
For digraphs, we defined stability only considering out-neighbourhoods. Nevertheless, an isomorphism $h$ between two digraphs not only maps the out-neighbourhood of a vertex $v$ bijectively to the out-neighbourhood of $h(v)$, but does the same with the in-neighbourhoods. So for the purpose of digraph isomorphism testing, the following stronger stability criterion is more useful: a $k$-colouring $\alpha$ of a digraph $G$ is {\em bi-stable} if for every pair of vertices $u,v\in V(G)$ with $\alpha(u)=\alpha(v)$ and every colour $c\in \{1,\ldots,k\}$, both $|N^+(u)\cap C^{\alpha}_c|=|N^+(v)\cap C^{\alpha}_c|$ and $|N^-(u)\cap C^{\alpha}_c|=|N^-(v)\cap C^{\alpha}_c|$ hold.

\begin{thm}
\label{thm:algmain_bidirected}
For any digraph $G$ on $n$ vertices with $m$ edges, in time $O((n+m)\log n)$ a canonical coarsest bi-stable colouring can be computed.
\end{thm}

\PF
Let $V=V(G)$. 
Construct a 2-edge coloured digraph $G'=(V,E,c)$ on the same vertex set as $G$ as follows: for every edge $(u,v)\in E(G)$, add an edge $e=(u,v)$ to $E$ with $c(e)=1$, and an edge $f=(v,u)$ to $E$ with $c(f)=2$. 
(Note that this may introduce parallel edges.)
Observe that a colouring $\alpha:V\to \{1,\ldots,k\}$ is edge-colour stable for $G'$ if and only if it is bi-stable for $G$, and that a canonical colouring method for $G'$ is a canonical colouring method for $G$. 
So Theorem~\ref{thm:algmain_edgecol} can be applied. We use that $G'$ has $2m\in O(n^2)$ edges, which yields the complexity bound $O((n+m)\log n)$. 
\QED

\section{Complexity Lower Bound}
\label{sec:lowerbound}

We shall prove our lower bound for undirected graphs; this makes 
it as general as possible.
The {\em cost} of a refining operation $(R,S)$ in a graph $G$ is 
\[
\cost(R,S):=|\{(u,v) \mid uv\in E(G),u\in R, v\in S\}|.
\]
This is basically the number of edges between $R$ and $S$, except that edges with both ends in $R\cap S$ are counted twice.  
For a partition $\pi$ that admits a refining operation $(R,S)$, denote by $\pi(R,S)$ the partition that results from this operation.
\begin{defi}
\label{defi:cost}
Let $G=(V,E)$ be a graph, and $\pi$ be a partition of $V$. 
\begin{itemize}
 \item If $\pi$ is stable, then $\cost(\pi):=0$.
 \item Otherwise, $\cost(\pi):=\min_{R,S} \cost(\pi(R,S)) + \cost(R,S)$, where the minimum is taken over all effective refining operations $(R,S)$ that can be applied to $\pi$.
\end{itemize}
\end{defi}

Note that this is well-defined; if $\pi$ is unstable, then there exists at least one effective elementary refining operation $(R,S)$, and for any such operation, $|\pi(R,S)|>|\pi|$.
We can now formulate the main result of this section.

\begin{thm}\label{thm:lowerbound}
For every integer $k\ge 2$, there is a graph $G_k$ with $n\in O(2^kk)$ vertices and $m \in O(2^kk^2)$ edges, such that $\cost(\alpha)\in \Omega((m+n)\log n)$, where $\alpha$ is the unit partition of $V(G_k)$.
\end{thm}

Note that this theorem implies a complexity lower bound for all partition-refinement based algorithms for colour refinement, as discussed in the introduction. 
We will first prove some basic observations related to the above definitions, then give the construction of the graph, and finally prove Theorem~\ref{thm:lowerbound}.

\subsection{Basic Observations}

We start with two basic properties of stable partitions. The first proposition follows easily from the definitions.
\begin{propo}
\label{propo:PiClosedSubgraphRemainsStable}
Let $G=(V,E)$ be a graph, and $\pi$ be a stable partition of $V$. For any $\pi$-closed subset $S\subseteq V$, $\pi[S]$ is a stable partition for $G[S]$. 
\end{propo}

\begin{propo}
\label{propo:distances}
Let $G=(V,E)$ be a graph, and $\pi$ be a stable partition of $V$. For any $\pi$-closed set $S$ and vertices $u,v\in V$: if the distance from $u$ to $S$ is different from the distance from $v$ to $S$, then $u\not\approx_{\pi} v$. 
\end{propo}

\PF 
Denote the distance from a vertex $x$ to $S$ by $\dist(x,S)$. 
W.l.o.g. we may assume that $\dist(u,S)<\dist(v,S)$, so in particular $\dist(u,S)$ is finite. 
We prove the statement by induction over $\dist(u,S)$.
If $\dist(u,S)=0$ then $u\in S$ but $v\not\in S$. Since $S$ is $\pi$-closed, this implies $u\not\approx_{\pi}v$. Otherwise, $u$ is adjacent to a vertex $w$ with $\dist(w,S)=\dist(u,S)-1$, but $v$ is not. Let $R\in \pi$ be the cell with $w\in R$. Then by induction, $|N(v)\cap R|=0$, so $u\not\approx_{\pi} v$, since $\pi$ is stable.
\QED

For a partition $\pi$ of $V$, denote by $\pi_{\infty}$ the coarsest stable partition of $V$ that refines $\pi$.

\begin{propo}
\label{propo:CostBoundKey}
Let $\pi$ and $\rho$ be partitions of $V$ such that $\pi\refby \rho\refby \pi_{\infty}$.
Then $\cost(\pi)\ge \cost(\rho)$.
\end{propo}

\PF
Let $(R,S)$ be a refining operation that can be applied to $\pi$, which yields $\pi'$. Then it can be observed that the operation $(R,S)$ can also be applied to $\rho$, and that for the resulting partition $\rho'$, it holds again that $\pi'\refby \rho'\refby \pi_{\infty}$ (Proposition~\ref{propo:RefOpSafe} shows that $\rho'\refby \pi_{\infty}$). 

An induction proof based on this observation shows that a minimum cost sequence of refining operations that refines $\pi$ to $\pi_\infty$ can also be applied to $\rho$, to yield the stable partition $\pi_\infty$, at the same cost. Therefore, $\cost(\pi)\ge \cost(\rho)$.
\QED

A refining operation $(R,S)$ on $\pi$ is {\em elementary} if both $R\in \pi$ and $S\in \pi$.
The next proposition shows that adding the word `elementary' in Definition~\ref{defi:cost} yields an equivalent definition.

\begin{propo}
\label{propo:ElementaryIrrelevant}
Let $\pi$ be an unstable partition of $V(G)$.
Then 
\[
\cost(\pi)=\min_{R,S} \cost(\pi(R,S)) + \cost(R,S),
\] 
where the minimum is taken over all effective elementary refining operations $(R,S)$ that can be applied to $\pi$.
\end{propo}
\PF
Let $(R,S)$ an nonelementary refining operation for $\pi$, and let $\rho_1$ be the result of applying $(R,S)$ to $\pi$. We shall prove that there is a sequence of elementary refining operations of total cost at most $\cost(R,S)$ that, when applied to $\pi$, yields a partition $\rho_2$ that refines $\rho_1$. The claim follows by Proposition~\ref{propo:CostBoundKey}. 

Suppose that $R$ consists of the cells $R_1,\ldots,R_q$ and $S$ consists of the cells $S_1,\ldots,S_p$. We apply the elementary refining operations $(R_i,S_j)$ for all $i\in \{1,\ldots,q\},j\in \{1,\ldots,p\}$ in an arbitrary order and let $\rho_2$ be the resulting partition. The cost of these elementary refinements is
\begin{align*}
\sum_{i,j}\cost(R_i,S_j)&=\sum_{i,j}|\{(u,v)\mid uv\in E(G), u\in R_i,v\in S_j\}| \\
&=|\{(u,v)\mid uv\in E(G), u\in R,v\in S\}|=\cost(R,S).
\end{align*}
It is easy to see that $\rho_2$ refines $\rho_1$. Indeed, if $u,v\in S$ belong to the same class of $\rho_2$, then they belong to the same class $S_j$, and for all classes $R_i$ they have the same number of neighbours in $R_i$. Hence they have the same number of neighbours in $R=\bigcup_iR_i$, and this means that they belong to the same class of $\rho_1$.
\QED

\subsection{Construction of the Graph}
\label{ssec:construction}

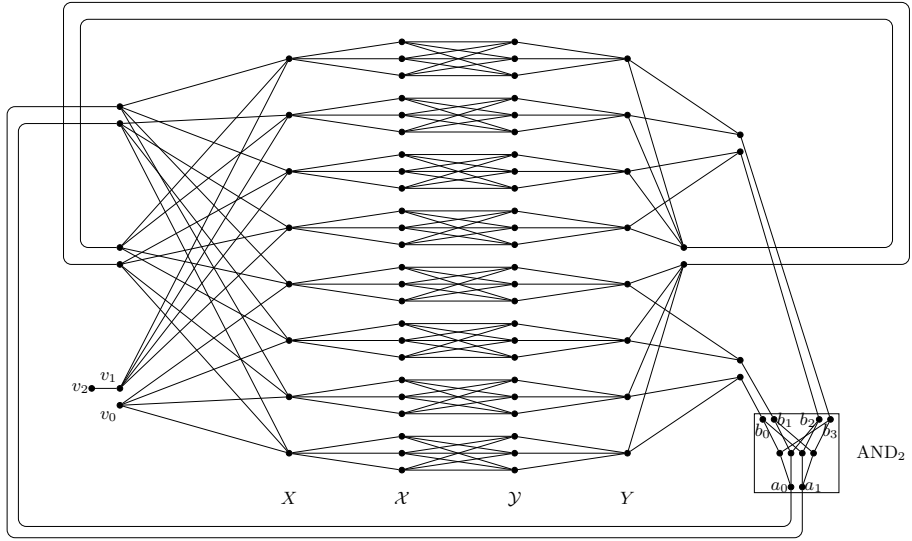
\begin{figure}
\resizebox{\textwidth}{!}{
\begin{tikzpicture}
[knoten/.style={circle,draw=black,fill=black,
	inner  sep=0pt,minimum  size=1mm}]
\foreach \xmax in {8} {

\foreach \x in {1,...,\xmax} \node[knoten] (x\x) at (2,\x) {};
\foreach \x in {1,...,\xmax} \foreach \y in {1,2,3} 
	\node[knoten] (x\x\y) at (4,\x+\y*0.3-2*0.3) {};
\foreach \x in {1,...,\xmax} \node[knoten] (y\x) at (8,\x) {};
\foreach \x in {1,...,\xmax} \foreach \y in {1,2,3} 
	\node[knoten] (y\x\y) at (6,\x+\y*0.3-2*0.3) {};
	
\foreach \x in {1,...,\xmax} \foreach \y in {1,2,3} {
	\draw[-] (x\x)--(x\x\y);
	\draw[-] (y\x)--(y\x\y);
	\foreach \z in {1,2,3} \draw[-] (x\x\y) -- (y\x\z);
};

\node[knoten] (bO1) at (-1,2-0.15) {};
\node[knoten] (bI1) at (-1,2+0.15) {};
\node[knoten] (bO2) at (-1,4.5-0.15) {};
\node[knoten] (bI2) at (-1,4.5+0.15) {};
\node[knoten] (bO3) at (-1,7-0.15) {};
\node[knoten] (bI3) at (-1,7+0.15) {};

\node[knoten] (start) at  (-1.5,2+0.15) {};
\node at (-1.5-0.2,2+0.15) {$v_2$};
\node  at (-1-.2,2-0.15-.2) {$v_0$};
\node at (-1-.2,2+0.15+.2) {$v_1$};

\draw[-] (start) -- (bI1);

\draw[-] (bO1) -- (x1);
\draw[-] (bO1) -- (x2);
\draw[-] (bO1) -- (x3);
\draw[-] (bO1) -- (x4);
\draw[-] (bI1) -- (x5);
\draw[-] (bI1) -- (x6);
\draw[-] (bI1) -- (x7);
\draw[-] (bI1) -- (x8);

\draw[-] (bO2) -- (x1);
\draw[-] (bO2) -- (x2);
\draw[-] (bO2) -- (x5);
\draw[-] (bO2) -- (x6);
\draw[-] (bI2) -- (x3);
\draw[-] (bI2) -- (x4);
\draw[-] (bI2) -- (x7);
\draw[-] (bI2) -- (x8);

\draw[-] (bO3) -- (x1);
\draw[-] (bO3) -- (x3);
\draw[-] (bO3) -- (x5);
\draw[-] (bO3) -- (x7);
\draw[-] (bI3) -- (x2);
\draw[-] (bI3) -- (x4);
\draw[-] (bI3) -- (x6);
\draw[-] (bI3) -- (x8);

\foreach \s in {1,2} {
	\node[knoten] (a1O\s) at (10,4*\s-1.5-0.15) {};
	\node[knoten] (a1I\s) at (10,4*\s-1.5+0.15) {};
	}
\node[knoten] (a2O) at (9,4.5-0.15) {};
\node[knoten] (a2I) at (9,4.5+0.15) {};

\draw[-] (a1O1) -- (y1);
\draw[-] (a1O1) -- (y2);
\draw[-] (a1I1) -- (y3);
\draw[-] (a1I1) -- (y4);
\draw[-] (a1O2) -- (y5);
\draw[-] (a1O2) -- (y6);
\draw[-] (a1I2) -- (y7);
\draw[-] (a1I2) -- (y8);

\foreach  \s in {1,...,4} \draw[-] (a2O) -- (y\s);
\foreach  \s in {5,...,8} \draw[-] (a2I) -- (y\s);

 \draw (10.25,1.5+.2) rectangle (11.75,0.5-.2);
\node at (12.5,1) {$\gadget_2$};
\node[knoten] (andO1) at (10.5-0.1,1.6) {};
\node[knoten] (andI1) at (10.5+0.1,1.6){};
\node[knoten] (andO2) at (11.5-0.1,1.6){};
\node[knoten] (andI2) at (11.5+0.1,1.6){};
\node[knoten] (andO3) at (11-0.1,0.4){};
\node[knoten] (andI3) at (11+0.1,0.4){};
\node[knoten] (CFI1) at (11-0.3,1){};
\node[knoten] (CFI2) at (11-0.1,1){};
\node[knoten] (CFI3) at (11+0.1,1){};
\node[knoten] (CFI4) at (11+0.3,1){};
\draw[-] (andO3) -- (CFI1);
\draw[-] (andO3) -- (CFI2);
\draw[-] (andI3) -- (CFI3);
\draw[-] (andI3) -- (CFI4);
\draw[-] (andO1) -- (CFI1);
\draw[-] (andO1) -- (CFI3);
\draw[-] (andI1) -- (CFI2);
\draw[-] (andI1) -- (CFI4);
\draw[-] (andO2) -- (CFI2);
\draw[-] (andO2) -- (CFI3);
\draw[-] (andI2) -- (CFI1);
\draw[-] (andI2) -- (CFI4);

\draw[-] (andO1) -- (a1O1);
\draw[-] (andI1) -- (a1I1);
\draw[-] (andO2) -- (a1O2);
\draw[-] (andI2) -- (a1I2);

\node at (10.5-0.1,1.6-.2) {$b_0$};
\node at (10.5+0.1+.2,1.6){$b_1$};
\node at (11.5-0.1-.2,1.6){$b_2$};
\node at (11.5+0.1,1.6-.2){$b_3$};

\node at (11-0.1-.2,0.4){$a_0$};
\node at (11+0.1+.2,0.4){$a_1$};

\foreach \hoch/\tief in {9/-.5} { 
\draw[-,rounded corners] (a2O) -- (13,4.5-0.15) -- (13,\hoch) -- (-2,\hoch) -- (-2,4.5-0.15) -- (bO2);
\draw[-,rounded corners] (a2I) -- (13-0.3,4.5+0.15) -- (13-0.3,\hoch-0.3) -- (-2+0.3,\hoch-0.3) -- (-2+0.3,4.5+0.15) -- (bI2);

\draw[-,rounded corners] (andI3) -- (11.1,\tief) -- (-3,\tief) -- (-3,7.15) -- (bI3);
\draw[-,rounded corners] (andO3) -- (11.1-0.2,\tief+0.2) -- (-3+0.2,\tief+0.2) -- (-3+0.2,7.15-0.3) -- (bO3);
}

\node at (2,0.2) {$\X$};
\node at (4,0.2) {$\XX$};
\node at (6,0.2) {$\YY$};
\node at (8,0.2) {$\Y$};
}
\end{tikzpicture}
}
\caption{The graph $G_3$}
\label{fig:G3}
\end{figure}

For $k\in \mathbb{N}$, denote $\blk_k=\{0,\ldots,2^k-1\}$. 
For $\ell\in \{0,\ldots,k\}$ and $q\in \{0,\ldots,2^{\ell}-1\}$, the subset $\blk^{\ell}_q=\{q2^{k-\ell},\ldots,(q+1)2^{k-\ell}-1\}$ is called the {\em $q$-th binary block of level $\ell$}.
Analogously, for any set of vertices with indices in $\blk_k$, we also consider binary blocks. For instance, if $X=\{x_i \mid i\in \blk_k\}$, then $X^{\ell}_q=\{x_i \mid i\in\blk^{\ell}_q\}$ is called a binary block of $X$. For such a set $X$, a {\em partition $\pi$ of $X$ into binary blocks} is a partition where every $S\in \pi$ is a binary block.
A key fact for binary blocks that we will often use is that for any $\ell$ and $q$,  $\blk^{\ell}_q=\blk^{\ell+1}_{2q}\cup \blk^{\ell+1}_{2q+1}$.

For every integer $k\ge 2$, we will construct a graph $G_k$. (An example for $k=3$ is given in Figure~\ref{fig:G3}.)
In its core this graph consists of the vertex sets $\X=\{x_i\mid i\in \blk_k\}$, $\XX=\{x^j_i\mid i\in \blk_k,j\in \{1,\ldots,k\}\}$,
$\YY=\{y^j_i\mid i\in \blk_k,j\in \{1,\ldots,k\}\}$ and $\Y=\{y_i\mid i\in \blk_k\}$. 
Every vertex $x_i$ is adjacent to $x^j_i$ for all $j\in\{1,\ldots,k\}$ and every $y_i$ is adjacent to all $y_i^j$. Furthermore, for all $i,j_1,j_2$ there is an edge between $x^{j_1}_i$ and $y^{j_2}_i$. 
(For $\XX$, {\em binary blocks} are subsets of the form $\XX^{\ell}_q:=\{x^j_i \mid i\in\blk^{\ell}_q, j\in\{1,\ldots,k\}\}$, and for $\YY$ the definition is analogous.)

We add gadgets to the graph to ensure that any sequence of refining operations behaves as follows.
After the first step, which distinguishes vertices according to their degrees, $\X$ and $\Y$ are cells of the resulting partition.
Next, $\X$ splits up into two binary blocks $\X^1_0$ and $\X^1_1$ of equal size. This causes $\XX$ to split up accordingly into $\XX^1_0$ and $\XX^1_1$. 
One of these cells will be used to halve $\YY$ in the same way. This refining operation $(R,S)$ is expensive because $[R,S]$ contains half of the edges between $\XX$ and $\YY$.
Next, $\Y$ can be split up into $\Y^1_0$ and $Y^1_1$. 
Once this happens, there is a gadget $\AND_1$ that causes the two cells $X^1_0$, $X^1_1$ to split up into the four cells $X^2_q$, for $q=0,\ldots,3$. 
Again, this causes cells in $\XX, \YY$ and $\Y$ to split up in the same way and to achieve this, half of the edges between $\XX$ and $\YY$ have to be considered. The next gadget $\AND_2$ ensures that if both cells of $\Y$ are split, then the four cells of $\X$ can be halved again, etc. 
In general, we design a gadget $\AND_{\ell}$ of level $\ell$ that ensures that 
if $\Y$ is partitioned into $2^{\ell+1}$ binary blocks of equal size, then $\X$ can be partitioned into $2^{\ell+2}$ binary blocks of equal size.
By halving all the cells of $\X$ and $\Y$ $k=\Theta(\log n)$ times (with $n=|V(G_k)|$), this refinement process ends up with a discrete colouring of these vertices. Since every iteration uses half of the edges between $\XX$ and $\YY$ (which are $\Theta(m)$), we get the cost lower bound of $\Omega(m\log n)$ (with $m=|E(G_k)|$).

\begin{figure}
\centering
\parbox{1.2in}{

\resizebox{!}{\pictureheight}{
\begin{tikzpicture}
[knoten/.style={circle,draw=black,fill=black,
	inner  sep=0pt,minimum  size=1mm},scale=2.5]

 \draw (10.25,1.5+.15) rectangle (11.75,0.5-.15);
\node[knoten] (andO1) at (10.5-0.1,1.6) [label=above:{\small$b_0$}] {};
\node[knoten] (andI1) at (10.5+0.1,1.6) [label=above:{\small$b_1$}]{};
\node[knoten] (andO2) at (11.5-0.1,1.6) [label=above:{\small$b_2$}]{};
\node[knoten] (andI2) at (11.5+0.1,1.6) [label=above:{\small$b_3$}]{};
\node[knoten] (andO3) at (11-0.1,0.4) [label=below:{\small$a_0$}]{};
\node[knoten] (andI3) at (11+0.1,0.4) [label=below:{\small$a_1$}]{};
\node[knoten] (CFI1) at (11-0.3,1) [label=left:{\small$c_0$}] {};
\node[knoten] (CFI2) at (11-0.1,1) [label=below:{\small$c_1$}]{};
\node[knoten] (CFI3) at (11+0.1,1) [label=below:{\small$c_2$}]{};
\node[knoten] (CFI4) at (11+0.3,1) [label=right:{\small$c_3$}]{};
\draw[-] (andO3) -- (CFI1);
\draw[-] (andO3) -- (CFI2);
\draw[-] (andI3) -- (CFI3);
\draw[-] (andI3) -- (CFI4);
\draw[-] (andO1) -- (CFI1);
\draw[-] (andO1) -- (CFI3);
\draw[-] (andI1) -- (CFI2);
\draw[-] (andI1) -- (CFI4);
\draw[-] (andO2) -- (CFI2);
\draw[-] (andO2) -- (CFI3);
\draw[-] (andI2) -- (CFI1);
\draw[-] (andI2) -- (CFI4);

\end{tikzpicture}
}

\caption{$\gadget_2$}
\label{fig:gadget2}}
\hspace{2cm}
\begin{minipage}{1.2in}

\resizebox{!}{\pictureheight}{
\begin{tikzpicture}
[knoten/.style={circle,draw=black,fill=black,
	inner  sep=0pt,minimum  size=1mm}]

 \draw (10.25,1.5+.15) rectangle (11.75,0.5-.15);
\node[knoten] (AndO1) at (10.5-0.1,1.6) [label=above:{\small$b_0$}]{};
\node[knoten] (AndI1) at (10.5+0.1,1.6){};
\node[knoten] (AndO2) at (11.5-0.1,1.6) {};
\node[knoten] (AndI2) at (11.5+0.1,1.6){};
\node[knoten] (AndO3) at (11-0.1,0.4){};
\node[knoten] (AndI3) at (11+0.1,0.4){};
\node[knoten] (CFI1) at (11-0.3,1){};
\node[knoten] (CFI2) at (11-0.1,1){};
\node[knoten] (CFI3) at (11+0.1,1){};
\node[knoten] (CFI4) at (11+0.3,1){};
\draw[-] (AndO3) -- (CFI1);
\draw[-] (AndO3) -- (CFI2);
\draw[-] (AndI3) -- (CFI3);
\draw[-] (AndI3) -- (CFI4);
\draw[-] (AndO1) -- (CFI1);
\draw[-] (AndO1) -- (CFI3);
\draw[-] (AndI1) -- (CFI2);
\draw[-] (AndI1) -- (CFI4);
\draw[-] (AndO2) -- (CFI2);
\draw[-] (AndO2) -- (CFI3);
\draw[-] (AndI2) -- (CFI1);
\draw[-] (AndI2) -- (CFI4);

\begin{scope}[xshift=2cm]
 \draw (10.25,1.5+.15) rectangle (11.75,0.5-.15);
\node[knoten] (aNdO1) at (10.5-0.1,1.6) {};
\node[knoten] (aNdI1) at (10.5+0.1,1.6){};
\node[knoten] (aNdO2) at (11.5-0.1,1.6) {};
\node[knoten] (aNdI2) at (11.5+0.1,1.6) [label=above:{\small$b_{7}$}]{};
\node[knoten] (aNdO3) at (11-0.1,0.4){};
\node[knoten] (aNdI3) at (11+0.1,0.4){};
\node[knoten] (CFI1) at (11-0.3,1){};
\node[knoten] (CFI2) at (11-0.1,1){};
\node[knoten] (CFI3) at (11+0.1,1){};
\node[knoten] (CFI4) at (11+0.3,1){};
\draw[-] (aNdO3) -- (CFI1);
\draw[-] (aNdO3) -- (CFI2);
\draw[-] (aNdI3) -- (CFI3);
\draw[-] (aNdI3) -- (CFI4);
\draw[-] (aNdO1) -- (CFI1);
\draw[-] (aNdO1) -- (CFI3);
\draw[-] (aNdI1) -- (CFI2);
\draw[-] (aNdI1) -- (CFI4);
\draw[-] (aNdO2) -- (CFI2);
\draw[-] (aNdO2) -- (CFI3);
\draw[-] (aNdI2) -- (CFI1);
\draw[-] (aNdI2) -- (CFI4);
\end{scope}

\begin{scope}[xshift=1cm,yshift=-2cm]
 \draw (10.25,1.5+.15) rectangle (11.75,0.5-.15);
\node[knoten] (anDO1) at (10.5-0.1,1.6) {};
\node[knoten] (anDI1) at (10.5+0.1,1.6){};
\node[knoten] (anDO2) at (11.5-0.1,1.6){};
\node[knoten] (anDI2) at (11.5+0.1,1.6){};
\node[knoten] (anDO3) at (11-0.1,0.4) [label=below:{\small$a_{0}$}]{};
\node[knoten] (anDI3) at (11+0.1,0.4){};
\node[knoten] (CFI1) at (11-0.3,1){};
\node[knoten] (CFI2) at (11-0.1,1){};
\node[knoten] (CFI3) at (11+0.1,1){};
\node[knoten] (CFI4) at (11+0.3,1){};
\draw[-] (anDO3) -- (CFI1);
\draw[-] (anDO3) -- (CFI2);
\draw[-] (anDI3) -- (CFI3);
\draw[-] (anDI3) -- (CFI4);
\draw[-] (anDO1) -- (CFI1);
\draw[-] (anDO1) -- (CFI3);
\draw[-] (anDI1) -- (CFI2);
\draw[-] (anDI1) -- (CFI4);
\draw[-] (anDO2) -- (CFI2);
\draw[-] (anDO2) -- (CFI3);
\draw[-] (anDI2) -- (CFI1);
\draw[-] (anDI2) -- (CFI4);
\end{scope}

\draw[-] (AndO3) -- (anDO1);
\draw[-] (AndI3) -- (anDI1);
\draw[-] (aNdO3) -- (anDO2);
\draw[-] (aNdI3) -- (anDI2);

\end{tikzpicture}
}
\caption{$\gadget_3$}
\label{fig:gadget3}
\end{minipage}
\end{figure}

We now define these gadgets in more detail.
For every integer $\ell\ge 1$, we define a gadget $\AND_{\ell}$, which consists of a graph $G$ together with two {\em out-terminals} $a_0,a_1$, and an ordered sequence of $p=2^{\ell}$ {\em in-terminals}  $b_0,\ldots,b_{p-1}$. 
For $\ell=1$, the graph $G$ has $V(G)=\{a_0,a_1,b_0,b_1\}$, and $E(G)=\{a_0b_0,a_1b_1\}$. 
For $\ell=2$, the graph $G$ is identical to the construction of Cai, F\"{u}rer and Immerman~\cite{CFI92}.
(See Figure \ref{fig:gadget2}. 
The out-terminals $a_0,a_1$ and in-terminals $b_0,\ldots,b_3$ are indicated.)
For $\ell\ge 3$, $\AND_{\ell}$ is obtained by taking one copy $G^*$ of an $\AND_2$-gadget,
and two copies $G'$ and $G''$ of an $\AND_{\ell-1}$-gadget, and adding four edges to connect the two pairs of in-terminals of $G^*$ with the pairs of out-terminals of $G'$ and $G''$, respectively.
As out-terminals of the resulting gadget we choose the out-terminals of $G^*$. The in-terminal sequence is obtained by concatenating the sequences of in-terminals of $G'$ and $G''$. 
(See Figure~\ref{fig:gadget3} for an example of $\AND_3$.)
For any $\AND_{\ell}$-gadget $G$ with in-terminals $b_0,\ldots,b_{2^{\ell}-1}$, the {\em in-terminal pairs} are pairs $b_{2p}$ and $b_{2p+1}$, for all $p\in \{0,\ldots,2^{\ell-1}-1\}$.

The graph $G_k$ is now constructed as follows. Start with vertex sets $X,\XX,\YY$ and $Y$, and edges between them, as defined above. For every $\ell\in \{1,\ldots,k-1\}$, we add a copy $G$ of an $\AND_{\ell}$-gadget to the graph. Denote the out- and in-terminals of $G$ by $a_0,a_1$ and $b_0,\ldots,b_{2^{\ell}-1}$, respectively. 
\begin{itemize}
\item 
For $i=0,1$ and all relevant $q$: we add edges from $a_i$ to every vertex in $X^{\ell+1}_{2q+i}$.
\item
For every $i$, we add edges from $b_i$ to every vertex in $Y^{\ell}_i$.
\end{itemize}
Finally, we add a {\em starting gadget} to the graph, consisting of three vertices $v_0,v_1,v_2$, the edge $v_1v_2$, and edges $\{v_0x_i \mid i\in\blk^1_0\}\cup \{v_1x_i \mid i\in\blk^1_1\}$.
See Figure~\ref{fig:G3} for an example of this construction. (In the figure, we have expanded the terminals of $\AND_2$ into edges, for readability. This does not affect the behaviour of the graph.)

\begin{propo}
\label{propo:Size}
$G_k$ has $O(2^kk)$ vertices and $O(2^kk^2)$ edges. 
\end{propo}

\PF
An easy induction proof shows that the $AND_{\ell}$-gadget has $O(2^{\ell})$ vertices and edges. So, all $\AND_{\ell}$ gadgets together, for $\ell\in \{1,\ldots,k-1\}$, have at most $O(2^k)$ vertices and edges. Therefore, the bounds on the total number of vertices and edges of $G_k$ are dominated by the number of vertices and edges in $G_k[\XX\cup \YY]$, which is $k2^{k+1}$ and $k^2 2^k$, respectively.
\QED

We now state and prove the key property for $AND_{\ell}$-gadgets. This requires the following definitions.
For a graph $G=(V,E)$, 
If $\psi$ is a partition of a {\em subset} $S\subseteq V$, then for short we say that a partition $\rho$ of $V$ {\em refines} $\psi$ if it refines $\psi\cup \{V\bs S\}$.
We say that $\rho$ {\em agrees with} $\psi$ if $\rho[S]=\psi$.
(So if $V\bs S\not=\emptyset$, one can choose $\rho$ such that it agrees with $\psi$ but does not refine $\psi$.)
For two graphs $G$ and $H$, by $G\disjunion H$ we denote the graph obtained by taking the disjoint union of $G$ and $H$.
We say that a partition $\pi$ of $V$ {\em distinguishes} two sets $V_1\subseteq V$ and $V_2\subseteq V$ if there is a set $R\in \pi$ with $|R\cap V_1|\not=|R\cap V_2|$.
This is used often for the case where $V_1=N(u)$ and $V_2=N(v)$ for two vertices $u$ and $v$, to conclude that if $\pi$ is stable, then $u\not\approx_{\pi} v$.
If $V_1=\{x\}$ and $V_2=\{y\}$, then we also say that {\em $\pi$ distinguishes $x$ from $y$}.

\begin{lem} \label{lem:ANDgadgets} 
  Let $G$ be an $\AND_{\ell}$-gadget with in-terminals $B=\{b_0,\ldots,b_{2^{\ell}-1}\}$ and out-terminals $a_0,a_1$. Let $\psi$ be a partition of $B$ into binary blocks, and let $\rho$ be the coarsest stable partition $\rho$ of $V(G)$ that refines $\psi$. Then $\rho$ agrees with $\psi$.
  Furthermore, $\rho$ distinguishes $a_0$ from $a_1$ if and only if $\psi$ distinguishes all in-terminal pairs.  
\end{lem}

\PF
We prove the statement by induction over $\ell$.
For $\ell=1$, the statement is trivial.
Now suppose $\ell=2$.
We only consider partitions of $\{b_0,\ldots,b_3\}$ into binary blocks. Because of the automorphisms of this gadget, it follows that it suffices to consider the following four partitions for $\psi$. For all of them, a corresponding partition $\rho$ is given; it can be verified that $\rho$ is the coarsest stable partition of
$V(AND_\ell)$
that refines $\psi$. 
(The nonterminal vertices are labeled $c_0,\ldots,c_3$, as shown in Figure \ref{fig:gadget2}.)
\begin{align*}
 \psi=\big\{ \{b_0,b_1,b_2,b_3\} \big\} &\Longrightarrow
 \rho=\psi\cup \big\{ \{c_0,c_1,c_2,c_3\},\{a_0,a_1\} \big\}, \\
 \psi=\big\{ \{b_0,b_1\},\{b_2,b_3\} \big\} &\Longrightarrow
 \rho=\psi\cup \big\{ \{c_0,c_1,c_2,c_3\},\{a_0,a_1\} \big\}, \\
 \psi=\big\{ \{b_0\},\{b_1\},\{b_2,b_3\} \big\} &\Longrightarrow
 \rho=\psi\cup \big\{ \{c_0,c_2\},\{c_1,c_3\},\{a_0,a_1\} \big\},\\
 \psi=\big\{ \{b_0\},\{b_1\},\{b_2\},\{b_3\} \big\}
 &\Longrightarrow
 \rho=\psi\cup \big\{ \{c_0\},\{c_1\},\{c_2\},\{c_3\},\{a_0\},\{a_1\} \big\}.
\end{align*}
We see that in all four cases, $\rho$ agrees with $\psi$ on $B$. Furthermore, $\rho$ distinguishes the out-terminals if and only if $\psi$ distinguishes all in-terminal pairs (which is only the case for the last $\psi$).

Now suppose $\ell\ge 3$. Recall that an $\AND_{\ell}$-gadget $H$ is obtained by taking two copies $G'$ and $G''$ of an $\AND_{\ell-1}$-gadget, and informally, putting a copy $G^*$ of an $\AND_2$-gadget on top of those. Any partition $\psi$ of the in-terminal set $B$ of $H$ into binary blocks corresponds to partitions $\psi'$ and $\psi''$ of the in-terminal sets $B'$ and $B''$ of $G'$ and $G''$ respectively, again into binary blocks.
So by induction, we have coarsest stable partitions $\rho'$ and $\rho''$ of $V(G')$ and $V(G'')$ that refine $\psi'$ and $\psi''$ and agree with them on $B'$ and $B''$, respectively. Together, this yields a partition $\pi$ of $V(G')\cup V(G'')$, which is stable for $G'\disjunion G''$, refines $\psi$, and agrees with $\psi$ on $B$. 
(To be precise: if $\psi$ is not the unit partition, then we can simply take $\pi=\rho'\cup \rho''$, because $\psi$ is a partition into binary blocks, and thus distinguishes every single in-terminal of $G'$ from every single in-terminal of $G''$. Otherwise, every set in $\pi$ should be the union of the two corresponding sets in $\rho'$ and $\rho''$.)
Then $\pi$ gives a partition of the out-terminals of $G'$ and $G''$,  which yields a matching partition $\psi^*$ of the in-terminals $B^*$ of $G^*$, again into binary blocks. Applying the induction hypothesis to $G^*$, we obtain a coarsest stable partition $\rho^*$ of $V(G^*)$ that refines and agrees with $\psi^*$.
Combining $\pi$ and $\rho^*$ yields a stable partition $\rho$ of the vertices $V(H)$ of the entire gadget.

Applying the induction hypothesis to $G'$ and $G''$ shows that at least one in-terminal pair of $G^*$ is not distinguished by $\psi^*$ if and only if at least one in-terminal pair of $G'$ or $G''$ is not distinguished by $\psi'$ or $\psi''$ respectively. Applying the induction hypothesis to $G^*$ then shows that $\rho$ does not distinguish the out-terminals of $H$ if $\psi$ does not distinguish at least one in-pair of $H$. This then also holds for the coarsest stable partition of $V(H)$ that refines $\psi$.

Finally, let $\psi$ be a partition into binary blocks of the in-terminals $B$ of $H$ that distinguishes every pair, and let $\rho$ be a coarsest stable partition that refines $\psi$. We prove that $\rho$ also distinguishes $a_0$ from $a_1$. 
By definition,
$\rho$ distinguishes any vertex from $B$ from any vertex not in $B$.
We conclude that for any two vertices $u,v\in V(H)$, if they have different distance to $B$, then $u\not\approx_{\rho} v$ (Proposition~\ref{propo:distances}).
So by Proposition~\ref{propo:PiClosedSubgraphRemainsStable}, 
$\rho$ induces stable partitions $\rho^*$ and $\pi$ for both $G^*$ and $G'\disjunion G''$, respectively. The graphs $G'$ and $G''$ are components of $G'\disjunion G''$, so we conclude that $\rho$ induces stable partitions $\rho'$ and $\rho''$ for both $G'$ and $G''$, respectively. 
By induction, it follows that $\rho'$ and $\rho''$ both distinguish the out-terminals of $G'$ and $G''$, respectively.
(If this holds for the coarsest stable partition, then it holds for any stable partition.)
Then 
$\psi:=\rho[B^*]$ (where $B^*$ denotes again the in-terminal set of $G^*$) distinguishes all in-terminal pairs of $G^*$. So by induction, $\rho$ distinguishes $a_0$ from $a_1$.
\QED

The following Corollary follows immediately from Lemma~\ref{lem:ANDgadgets}.

\begin{corol}
\label{cor:ANDgadget}
Let $\pi$ be a stable partition for an AND-gadget $G$ such that $\psi=\pi[B]$ is a partition of the in-terminals $B$ into binary blocks,
and such that $B$ is $\pi$-closed. 
If $\pi$ does not distinguish the out-terminals, then at least one in-terminal pair is not distinguished.
\end{corol}
\PF 
Since $B$ is $\pi$-closed, $\pi$ refines $\psi=\pi[B]$. 
Since $\pi$ is stable, it refines the coarsest stable partition $\rho$ of $V(G)$ that refines $\psi$. Now apply Lemma~\ref{lem:ANDgadgets}.\QED

\subsection{Cost Lower Bound Proof}

Intuitively, at level $\ell$ of the refinement process, the current partition contains all blocks $\XX^{\ell+1}_q$ of level $\ell+1$ and for all $0\leq q < 2^{\ell}$, either $\YY^{\ell}_q$ or the two blocks $\YY^{\ell+1}_{2q}$ and $\YY^{\ell+1}_{2q+1}$.
In this situation one can split up the blocks $\YY^{\ell}_q$ into blocks $\YY^{\ell+1}_{2q}$ and $\YY^{\ell+1}_{2q+1}$ using either refining operation $(\XX^{\ell+1}_{2q},\YY^{\ell}_q)$ or $(\XX^{\ell+1}_{2q+1},\YY^{\ell}_q)$. 
These operations both have cost $2^{k-(\ell+1)}k^2$, and refining all the $\YY^{\ell}_q$ cells in this way costs $2^{k-1}k^2$. 
Once $\YY$ is partitioned into binary blocks of level $\ell+1$, we can partition $\XX$ into blocks of level $\ell+2$ (using the $\gadget_{\ell}$-gadget), and proceed the same way. 
Since there are $k$ such refinement levels, we can lower bound the total cost of refining the graph by $2^{k-1}k^3 =\Omega(m\log n)$ and are done. What remains to show is that applying the refining operations in this specific way is the only way to obtain a stable partition. 
To formalise this, we introduce a number of partitions of $V(G_k)$ that are stable with respect to the (spanning) subgraph $G'_k=G_k-[\XX,\YY]$, and that partition $\XX$ and $\YY$ into binary blocks. (For disjoint vertex sets $S$, $T$, we denote $[S,T]=\{uv\in E(G) \mid u\in S, v\in T\}$.)
So on $G_k$, these partitions can only be refined using operations $(R,S)$, where $R$ is a binary block of $\XX$ and $S$ is a binary block of $\YY$. 

\begin{defi}
\label{defi:QlStable}
For any $\ell\in \{0,\ldots,k-1\}$, and nonempty set $Q\subseteq \blk_{\ell}$, by $\tau_{Q,\ell}$ we denote the partition of $\XX\cup \YY$ that contains cells 
\begin{itemize}
 \item 
 $\XX^{\ell+1}_q$ for all $q\in \blk_{\ell+1}$, 
 \item
 $\YY^{\ell}_q$ for all $q\in Q$, and 
both  $\YY^{\ell+1}_{2q}$ and $\YY^{\ell+1}_{2q+1}$ for all $q\in \blk_{\ell}\bs Q$.
\end{itemize}
$\pi_{Q,\ell}$ denotes the coarsest stable partition for $G'_k=G_k-[\XX,\YY]$ that refines $\tau_{Q,\ell}$.
\end{defi}

We now show that for every $\ell$ and $Q$, there is also a stable partition of $G'_k$ that partitions $\XX$ and $\YY$ as prescribed by the above definition. In particular, this holds for $\pi_{Q,\ell}$.

\begin{lem}
\label{lem:QlStableExists}
For every $\ell\in \{0,\ldots,k-1\}$ and nonempty set $Q\subseteq \blk_{\ell}$, $\pi_{Q,\ell}$ agrees with $\tau_{Q,\ell}$.
\end{lem}

\PF
We design a stable partition $\rho$ of $V(G_k)=V(G'_k)$ that is stable on $G'_k$, and agrees with $\tau_{Q,\ell}$. So we start with $\rho=\tau_{Q,\ell}$. For every cell $\XX^{\ell+1}_q$ in $\tau_{Q,\ell}$, we add the cell $X^{\ell+1}_q$ to $\rho$. For every cell $\YY^m_q$ in $\tau_{Q,\ell}$ ($\ell\le m\le \ell+1$), we add the cell $Y^m_q$ to $\rho$. 
Then we add cells $\{v_0\}$, $\{v_1\}$ and $\{v_2\}$.

For every $\AND_p$-gadget $G$ of $G_k$ (with in-terminals adjacent to $Y$ and out-terminals adjacent to $X$), we define a partition $\psi$ of the in-terminals $B$ as follows: for $u,v\in B$, $u\approx_{\psi} v$ if and only if $N(u)\cap Y$ is not distinguished from $N(v)\cap Y$. Note that this yields a partition of $B$ into binary blocks, and that this distinguishes an in-terminal pair $b_{2q}$, $b_{2q+1}$ (which are adjacent to $Y^p_{2q}$ and $Y^p_{2q+1}$, respectively, with union $Y^{p-1}_q$) if and only if $\ell\ge p$ holds, or both $\ell=p-1$ and $q\not\in Q$ hold.
Now we extend $\rho$ by adding all cells of the coarsest stable partition of the $\AND_p$-gadget $G$ that refines $\psi$.
By Lemma~\ref{lem:ANDgadgets}, this partition distinguishes the out-terminals of $G$ if and only if $\ell\ge p$ (since $Q$ is nonempty). Extending $\rho$ this way for every AND-gadget yields the final partition $\rho$ of $V(G_k)$.  By definition, $\rho$ agrees with $\tau_{Q,\ell}$. From the construction, the stability condition is easily verified for almost all cells of $\rho$. Only cells $\{a_0,a_1\}\in \rho$ consisting of out-terminals of $\AND_p$-gadgets need to be considered in more detail. As noted before, such cells only occur when $p\ge \ell+1$. Then we have for every integer $q$ that $X^{p+1}_{2q}\cup X^{p+1}_{2q+1}=X^p_q\subseteq X^{\ell+1}_{q'}\in \rho$ (for some value $q'$). Since $a_0$ is adjacent to every $X^{p+1}_{2q}$ and $a_1$ is adjacent to every $X^{p+1}_{2q+1}$, it follows that $N(a_0)$ and $N(a_1)$ are not distinguished by $\rho$. Therefore, $\rho$ is stable for $G'_k$. Then the coarsest stable partition $\pi_{Q,\ell}$ that refines $\tau_{Q,\ell}$ also agrees with $\tau_{Q,\ell}$. 
\QED

Since $\pi_{Q,\ell}$ is stable on $G'_k$, any effective refining operation (with respect to $G_k$) should involve the edges between $\XX$ and $\YY$. 
Since $\pi_{Q,\ell}$ partitions $\XX$ and $\YY$ as prescribed by $\tau_{Q,\ell}$, we conclude that 
any effective elementary refining operation has the form described in the following corollary. Recall that a refining operation $(R,S)$ for a partition $\pi$ is \emph{elementary} if both $R$ and $S$ are classes of $\pi$, and that by Proposition~\ref{propo:ElementaryIrrelevant} it suffices to consider elementary refining operations.

\begin{corol}
\label{cor:OnlySpecificOperations}
Let $(R,S)$ be an effective elementary refining operation on $\pi_{Q,\ell}$. Then for some $q\in Q$, $R=\XX^{\ell+1}_{2q}$ or $R=\XX^{\ell+1}_{2q+1}$, and $S=\YY^{\ell}_q$. The cost of this operation is $k^2 2^{k-(\ell+1)}$.
\end{corol}

This motivates the following definition: for $q\in Q$, by $r_q(\pi_{Q,\ell})$ we denote the partition of $V(G_k)$ that results from the above refining operation. (Both choices of $R$ yield the same result.)

\begin{lem}
\label{lem:PartialOrder_Qlpartitions} 
For every $\ell\in \{0,\ldots,k-1\}$, nonempty $Q\subseteq \blk_\ell$ and $q\in Q$:
\begin{itemize}
 \item $r_q(\pi_{Q,\ell})\refby \pi_{\blk_{\ell+1},\ell+1}$, and
 \item if $Q'=Q\bs \{q\}$ is nonempty, then $r_q(\pi_{Q,\ell})\refby \pi_{Q',\ell}$.
\end{itemize}
\end{lem}

\PF  
Choose $Q'$ and $\ell'$ satisfying one of the conditions (i.e. $Q'=\blk_{\ell+1}$ and $\ell'=\ell+1$, or $Q'=Q\bs \{q\}$ and $\ell'=\ell$).
Then $\tau_{Q,\ell}\refby \tau_{Q',\ell'}$, so also $\pi_{Q,\ell}\refby \pi_{Q',\ell'}$ (since $\pi_{Q',\ell'}$ is also a stable partition that refines $\tau_{Q,\ell}$).
If we now obtain a partition $\rho$ from $\pi_{Q,\ell}$ by splitting up one cell such that the only vertex pairs $u,v$ with $u\approx_{\pi_{Q,\ell}} v$ but $u\not\approx_{\rho} v$ are vertex pairs with $u\not\approx_{\pi_{Q',\ell'}} v$, then clearly still $\rho\refby \pi_{Q',\ell'}$ holds. This is exactly how $r_q(\pi_{Q,\ell})$ is obtained.
\QED

\begin{lem}
\label{lem:StableDiscrete}
Let $\omega$ be the coarsest stable partition for $G_k$.
For all $\ell\in \{0,\ldots,k-1\}$ and nonempty $Q\subseteq \blk_{\ell}$: $\pi_{Q,\ell}\refby \omega$.
\end{lem}

\PF  
First, we note that by considering the various vertex degrees and using Proposition~\ref{propo:distances}, one can verify that $\omega$ refines $\{X,\XX,\YY,\Y,\{v_0\},\{v_1\},\{v_2\},V_G\}$, where $V_G$ denotes all vertices in AND-gadgets. 
In particular, $V_G$ is $\omega$-closed, so $\omega$ induces a stable partition on $G[V_G]$ (Proposition~\ref{propo:PiClosedSubgraphRemainsStable}), and therefore it does so on every AND-gadget of $G_k$ (which are components of $G[V_G]$).
Note that for any two different $\AND_{\ell}$-gadgets $H_1$ and $H_2$ of $G_k$, there exists an integer $d$ such that $H_1$ contains a vertex at distance exactly $d$ from the $\omega$-closed set $X\cup Y$, but $H_2$ does not.
This observation can be combined with Proposition~\ref{propo:distances} to show that if $u$ and $v$ are part of different AND-gadgets, then $u\not\approx_{\omega} v$. Subsequently this yields that for any AND-gadget of $G_k$ with output terminals $a_0,a_1$, the set $\{a_0,a_1\}$ is $\omega$-closed,
and the set of input terminals $B$ of this gadget is $\omega$-closed.

We now prove that $\omega[X]$ is discrete.
Suppose that there is an AND-gadget in $G_k$ for which the out-terminals are not distinguished by $\omega$. Then let $\ell$ be the minimum value such that this holds for the $\AND_{\ell}$-gadget $G$ of $G_k$.
As observed above, we may apply Corollary~\ref{cor:ANDgadget} to $G$, 
which shows that there is at least one pair of in-terminals $b_{2q}$ and $b_{2q+1}$ that is not distinguished by $\omega$.
By stability, and since $Y$ is $\omega$-closed, 
this shows that there are vertices $y_i\in Y^{\ell}_{2q}$ and $y_j\in Y^{\ell}_{2q+1}$ in the adjacent binary blocks such that $y_i\approx_{\omega} y_j$.
Then, considering the $\omega$-closed subgraph $G_k[X\cup \XX\cup \YY\cup Y]$, it easily follows that $x_i\approx_{\omega} x_j$.
If $\ell\ge 2$, then $x_i\in X^{\ell}_{2q}$ is adjacent to the out-terminal $a_0$ of the $\AND_{\ell-1}$ gadget of $G_k$, whereas $x_j\in X^{\ell}_{2q+1}$ is adjacent to the other out-terminal $a_1$ of this gadget. By choice of $\ell$, $a_0\not\approx_{\omega} a_1$, so since $\{a_0,a_1\}$ is $\omega$-closed, this gives a contradiction with stability. If $\ell=1$, then we consider the starting gadget: $x_i\in X^{\ell}_{0}$ is adjacent to $v_0$, and $x_j\in X^{\ell}_{1}$ is adjacent to $v_1$, but $\{v_0\}$ and $\{v_1\}$ are distinct cells of $\omega$, a contradiction with stability.

We conclude that for every AND-gadget, $\omega$ distinguishes the out-terminals. Since every vertex $x_i\in X$ is adjacent to a unique set of such out-terminals, it follows that $\omega[X]$ is discrete. Therefore, for every $q$, $\XX^k_q$ and $\YY^k_q$ are $\omega$-closed.
Hence $\omega$ refines $\tau_{Q,\ell}$ for every $Q$ and $\ell$, and thus it refines $\pi_{Q,\ell}$ for every $Q$ and $\ell$.
\QED

\PFof Theorem~\ref{thm:lowerbound}: 
Let $G_k$ be the graph described in Section~\ref{ssec:construction}, and $\pi_{Q,\ell}$ be the partitions of $V(G_k)$ from Definition~\ref{defi:QlStable}. 
By Lemma~\ref{lem:StableDiscrete}, the coarsest stable partition $\omega$ of $G$ refines all partitions $\pi_{Q,\ell}$.
For ease of notation, we define $\pi_{\emptyset,\ell}:=\pi_{\blk_{\ell+1},\ell+1}$ for all $\ell<k-1$.
By Corollary~\ref{cor:OnlySpecificOperations}, any effective elementary refining operation on a partition $\pi_{Q,\ell}$ has cost $2^{k-(\ell+1)}k^2$, and results in $r_q(\pi_{Q,\ell})$ for some $q\in Q$. Denote $Q'=Q\bs \{q\}$.
By Lemma~\ref{lem:PartialOrder_Qlpartitions}, $r_q(\pi_{Q,\ell})\refby \pi_{Q',\ell}$.
By Proposition~\ref{propo:ElementaryIrrelevant}, to compute the $\cost(\pi_{Q,\ell})$, it suffices to consider only partitions that can be obtained by {\em elementary} refining operations. 
So we may now apply Proposition~\ref{propo:CostBoundKey} to conclude that 
\[
\cost(\pi_{Q,\ell})\geq 2^{k-(\ell+1)}k^2 + \min_{q\in Q}\cost(\pi_{Q\bs\{q\},\ell}).
\]
By induction on $|Q|$ it then follows that $\cost(\pi_{\blk_{\ell},\ell})\geq 2^{k-1}k^2+\cost(\pi_{\blk_{\ell+1},\ell+1})$ for all $0\leq \ell\le k-1$. Hence, by induction on $\ell$, $\cost(\pi_{\blk_{0},0})\geq 2^{k-1}k^3$, which lower bounds $\cost(\alpha)$.
By Proposition~\ref{propo:Size}, $n\in O(2^kk)$ and $m\in O(2^kk^2)$,
so $\log n\in O(k)$. This shows that $\cost(\alpha)\in \Omega((m+n)\log n)$.
\QED

\subsection{Related lower bounds}\label{sec:bisim}

In this section, we sketch how our construction also yields lower bounds for two other partitioning problems.

\subsubsection{Bisimilarity}
Bisimilarity is a key concept in concurrency theory and automated verification.
A bisimulation is a binary relation defined on the states of a transition system (or between two transition systems). Intuitively, two states are bisimilar if the processes starting in these states look the same. Formally, a \emph{transition system} is a vertex-labelled directed graph. Let $S=(V,E,\lambda)$, where $(V,E)$ is a directed graph and $\lambda$ a function that assigns a set of properties to each state $v\in V$. A \emph{bisimulation} on $S$ is a relation $\sim$ on $V$ satisfying the following three properties for all $v,w\in V$ such that $v\sim w$:
\begin{enumerate}[label=(\roman*),leftmargin=2.5em]
  \item $\lambda(v)=\lambda(w)$;
  \item for all $v'\in N^+(v)$ there is a $w'\in N^+(w)$ such that $v'\sim w'$;
  \item for all $w'\in N^+(w)$ there is a $v'\in N^+(v)$ such that $v'\sim w'$.
\end{enumerate}
Not every bisimulation is an equivalence relation, but the reflexive
symmetric transitive closure of a bisimulation is still a
bisimulation.
For convenience, in the following we assume that all
bisimulations are equivalence relations. This is justified by the fact that the partition refinement algorithms (see below) that are commonly used to compute bisimulations, and that we study here, represent the relations using partitions of the vertex set and hence implicitly assume that the relations they represent are equivalence relations.

It is not hard to see that on each transition system $S$ there is a
unique coarsest bisimulation, which we call the \emph{bisimilarity relation} on $S$. The bisimilarity relation can be defined by letting $v$ be \emph{bisimilar} to $w$ if there is a bisimulation $\sim$ such that $v\sim w$; it is then straightforward to verify that bisimilarity is a bisumlation and that all other bisimulations refine it.
We remark that the bisimilarity relation on a transition system is precisely what Paige and Tarjan~\cite{paitar87} call the \emph{coarsest relational partition} of the initial partition given by the labelling. Thus the problem of computing the bisimilarity relation of a given transition system is equivalent to the problem of computing the 
coarsest relational partition considered in \cite{paitar87}.

Note the similarity between a bisimulation and a stable colouring of a
vertex-coloured digraph, which we may view as a transition system with
a labelling $\lambda$ that maps each vertex to its colour. Condition
(i) just says that a bisimulation refines the original colouring, as a
stable colouring is supposed to do as well. Conditions (ii) and (iii),
which are equivalent under the assumption that a bisimulation be an
equivalence relation and hence symmetric, says that if two vertices
$v,w$ are in the same class $C$ then for every other class $D$, either
both $v$ and $w$ have an out-neighbour in $D$ or neither of them has.
Thus instead of refining by the degree in $D$, we just refine by the
Boolean value ``degree at least $1$''. This immediately implies that
the coarsest stable colouring of $S$ refines the coarsest
bisimulation, that is, the bisimilarity
relation, on $S$.

It should be clear from these considerations that the bisimilarity
relation on a transition system $S$ with $n$ vertices and $m$ edges
can be computed in time $O((n+m)\log n)$ by a slight modification of the
partitioning algorithm for computing the coarsest stable colouring
(assuming, of course, that the labels can be computed and compared in
constant time)  \cite{paitar87}.

As for the coarsest stable colouring, we may ask if the bisimilarity
relation can be computed in linear time. It turns out that our lower
bound for colour refinement implies a lower bound for bisimilarity.
Again, we consider the class of \emph{partition refinement algorithms}. As the partition refinement algorithms for colour refinement, partition refinement algorithms for bisimilarity maintain a partition of the set of vertices of the given transition system, and they iteratively refine it using refining operations until a bisimulation is reached. In each {\em refining operation}, such an algorithm chooses a union of current partition cells as {\em refining set} $R$, and chooses another (possibly overlapping) union of partition cells $S$. Cells in $S$ are split up according to the out-neighbourhoods of the vertices in the cells in $R$. That is, two vertices $v,w$ currently in the same cell in $S$ remain in the same cell after the refinement step if and only if for all cells 
$C$ of the partition, with $C\subseteq R$, it holds that
\[
N^+(v)\cap C\neq\emptyset\iff N^+(v')\cap C\neq\emptyset.
\]
Recall that $N^+(v)$ denotes the set of out-neighbours of a vertex $v$ in a directed graph (or transitition system).
The \emph{cost} $\bcost(R,S)$ of such a refinement relation $(R,S)$ is the number of edges from $S$ to $R$. Again, the sum of the costs of all refinement operations is a reasonable lower bound for the running time of a partition refinement algorithm. The \emph{cost} $\bcost(\alpha)$ of a partition $\alpha$ of  the vertex set is then defined as the minimum cost of a sequence of refinement operations that transforms $\alpha$ to the coarsest bisimulation refining it (see Definition~\ref{defi:cost}).

\begin{thm}\label{thm:bisimilarity} For every integer $k\ge 2$, there
  is a transition system $S_k$ with $n\in O(2^kk)$ vertices and
  $m \in O(2^kk^2)$ edges and constant labelling function, such that 
such that $\bcost(\alpha)\in \Omega((m+n)\log n)$, where $\alpha$ is the unit partition of $V(S_k)$.
\end{thm}

\begin{proof}[sketch]
The proof is essentially the same as the proof of
Theorem~\ref{thm:lowerbound}. The transition system $S_k$ is a
directed version of the the graph $G_k$. Figure~\ref{fig:S3}
illustrates the direction of the edges. All vertices get the same
label.

It is not hard to show that the bisimilarity classes of $S_k$ are
exactly the same as the colour classes of $G_k$ in the coarsest stable
colouring and that essentially the refinement steps do the same on
$G_k$ and $S_k$. Thus the lower-bound proof carries
over.
\qed
\end{proof}

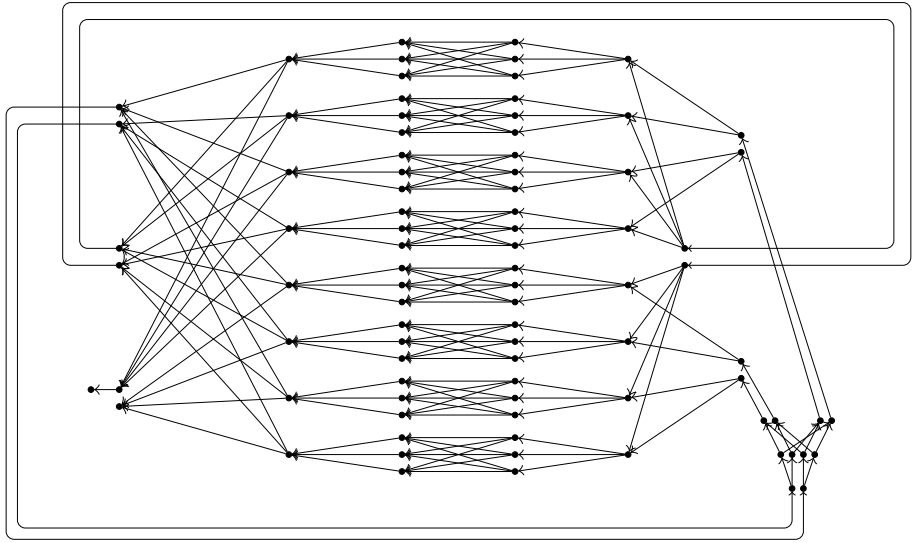
\begin{figure}
\resizebox{\textwidth}{!}{
\begin{tikzpicture}
[knoten/.style={circle,draw=black,fill=black,
	inner  sep=0pt,minimum  size=1mm},
      decoration={
        markings,
        mark=at position 1 with {\arrow[scale=1.5]{>}};
      },
]
\foreach \xmax in {8} {

\foreach \x in {1,...,\xmax} \node[knoten] (x\x) at (2,\x) {};
\foreach \x in {1,...,\xmax} \foreach \y in {1,2,3} 
	\node[knoten] (x\x\y) at (4,\x+\y*0.3-2*0.3) {};
\foreach \x in {1,...,\xmax} \node[knoten] (y\x) at (8,\x) {};
\foreach \x in {1,...,\xmax} \foreach \y in {1,2,3} 
	\node[knoten] (y\x\y) at (6,\x+\y*0.3-2*0.3) {};
	
\foreach \x in {1,...,\xmax} \foreach \y in {1,2,3} {
	\draw[postaction={decorate}](x\x\y)--(x\x);
	\draw[postaction={decorate}] (y\x)--(y\x\y);
	\foreach \z in {1,2,3} \draw[postaction={decorate}](y\x\z)--(x\x\y);
};

\node[knoten] (bO1) at (-1,2-0.15) {};
\node[knoten] (bI1) at (-1,2+0.15) {};
\node[knoten] (bO2) at (-1,4.5-0.15) {};
\node[knoten] (bI2) at (-1,4.5+0.15) {};
\node[knoten] (bO3) at (-1,7-0.15) {};
\node[knoten] (bI3) at (-1,7+0.15) {};

\node[knoten] (start) at  (-1.5,2+0.15) {};

\draw[postaction={decorate}] (bI1) -- (start);

\draw[postaction={decorate}](x1)--(bO1);
\draw[postaction={decorate}](x2)--(bO1);
\draw[postaction={decorate}](x3)--(bO1);
\draw[postaction={decorate}](x4)--(bO1);
\draw[postaction={decorate}](x5)--(bI1);
\draw[postaction={decorate}](x6)--(bI1);
\draw[postaction={decorate}](x7)--(bI1);
\draw[postaction={decorate}](x8)--(bI1);

\draw[postaction={decorate}](x1)--(bO2);
\draw[postaction={decorate}](x2)--(bO2);
\draw[postaction={decorate}](x5)--(bO2);
\draw[postaction={decorate}](x6)--(bO2);
\draw[postaction={decorate}](x3)--(bI2);
\draw[postaction={decorate}](x4)--(bI2);
\draw[postaction={decorate}](x7)--(bI2);
\draw[postaction={decorate}](x8)--(bI2);

\draw[postaction={decorate}](x1)--(bO3);
\draw[postaction={decorate}](x3)--(bO3);
\draw[postaction={decorate}](x5)--(bO3);
\draw[postaction={decorate}](x7)--(bO3);
\draw[postaction={decorate}](x2)--(bI3);
\draw[postaction={decorate}](x4)--(bI3);
\draw[postaction={decorate}](x6)--(bI3);
\draw[postaction={decorate}](x8)--(bI3);

\foreach \s in {1,2} {
	\node[knoten] (a1O\s) at (10,4*\s-1.5-0.15) {};
	\node[knoten] (a1I\s) at (10,4*\s-1.5+0.15) {};
	}
\node[knoten] (a2O) at (9,4.5-0.15) {};
\node[knoten] (a2I) at (9,4.5+0.15) {};

\draw[postaction={decorate}] (a1O1) -- (y1);
\draw[postaction={decorate}] (a1O1) -- (y2);
\draw[postaction={decorate}] (a1I1) -- (y3);
\draw[postaction={decorate}] (a1I1) -- (y4);
\draw[postaction={decorate}] (a1O2) -- (y5);
\draw[postaction={decorate}] (a1O2) -- (y6);
\draw[postaction={decorate}] (a1I2) -- (y7);
\draw[postaction={decorate}] (a1I2) -- (y8);

\foreach  \s in {1,...,4} \draw[postaction={decorate}] (a2O) -- (y\s);
\foreach  \s in {5,...,8} \draw[postaction={decorate}] (a2I) -- (y\s);

\node[knoten] (andO1) at (10.5-0.1,1.6) {};
\node[knoten] (andI1) at (10.5+0.1,1.6){};
\node[knoten] (andO2) at (11.5-0.1,1.6){};
\node[knoten] (andI2) at (11.5+0.1,1.6){};
\node[knoten] (andO3) at (11-0.1,0.4){};
\node[knoten] (andI3) at (11+0.1,0.4){};
\node[knoten] (CFI1) at (11-0.3,1){};
\node[knoten] (CFI2) at (11-0.1,1){};
\node[knoten] (CFI3) at (11+0.1,1){};
\node[knoten] (CFI4) at (11+0.3,1){};
\draw[postaction={decorate}] (andO3) -- (CFI1);
\draw[postaction={decorate}] (andO3) -- (CFI2);
\draw[postaction={decorate}] (andI3) -- (CFI3);
\draw[postaction={decorate}] (andI3) -- (CFI4);
\draw[postaction={decorate}](CFI1)--(andO1);
\draw[postaction={decorate}](CFI3)--(andO1);
\draw[postaction={decorate}](CFI2)--(andI1);
\draw[postaction={decorate}](CFI4)--(andI1);
\draw[postaction={decorate}](CFI2)--(andO2);
\draw[postaction={decorate}](CFI3)--(andO2);
\draw[postaction={decorate}](CFI1)--(andI2);
\draw[postaction={decorate}](CFI4)--(andI2);

\draw[postaction={decorate}] (andO1) -- (a1O1);
\draw[postaction={decorate}] (andI1) -- (a1I1);
\draw[postaction={decorate}] (andO2) -- (a1O2);
\draw[postaction={decorate}] (andI2) -- (a1I2);

\foreach \hoch/\tief in {9/-.5} { 
\draw[<-,rounded corners] (a2O) -- (13,4.5-0.15) -- (13,\hoch) -- (-2,\hoch) -- (-2,4.5-0.15) -- (bO2);
\draw[<-,rounded corners] (a2I) -- (13-0.3,4.5+0.15) -- (13-0.3,\hoch-0.3) -- (-2+0.3,\hoch-0.3) -- (-2+0.3,4.5+0.15) -- (bI2);

\draw[<-,rounded corners] (andI3) -- (11.1,\tief) -- (-3,\tief) -- (-3,7.15) -- (bI3);
\draw[<-,rounded corners] (andO3) -- (11.1-0.2,\tief+0.2) -- (-3+0.2,\tief+0.2) -- (-3+0.2,7.15-0.3) -- (bO3);
}

}
\end{tikzpicture}
}
\caption{The transitions system $S_3$ corresponding to the graph $G_3$
of Figure~\ref{fig:G3}}
\label{fig:S3}
\end{figure}

\subsubsection{Equivalence in 2-Variable Logic}

It is a well-known fact (due to Immerman and Lander~\cite{immlan90}) that colour refinement assigns the same colour to two vertices of a graph if and only if the vertices satisfy the same formulas of the logic $\LC^2$, \emph{two-variable first-order logic with counting}. 

\emph{Two variable first-order logic} $\LL^2$ is the fragment of first order logic consisting of all formulas built with just two variables. For example, the following $\LL^2$-formula $\phi(x)$ in the language of directed graphs says that from vertex $x$ one can reach a sink (a vertex of out-degree $0$) in four steps:
\[
\phi(x):=\exists y(E(x,y)\wedge\exists x(E(y,x)\wedge\exists y(E(x,y)\wedge\exists x(E(y,x)\wedge\forall y\,\neg E(x,y))))).
\]
\emph{Two variable first-order logic with counting} $\LC^2$ is the extension of $\LL^2$ by \emph{counting quantifiers} of the form $\exists^{\ge i}x$, for all $i\ge 1$. For example, the following $\LC^2$-formula $\psi(x)$ in the language of directed graphs says that from vertex $x$ one can reach a vertex of out-degree at least $10$ in four steps:
\[
\psi(x):=\exists y(E(x,y)\wedge\exists x(E(y,x)\wedge\exists y(E(x,y)\wedge\exists x(E(y,x)\wedge\exists^{\ge10}y E(x,y))))).
\]
This formula is not equivalent to any formula of $\LL^2$. Two-variable logics, and more generally finite variable logics, have been studied extensively in finite model theory (see, for example, \cite{ebbflu99,imm99,lib04,grakollib+07}).

We call two vertices of a graph \emph{$\LL^2$-equivalent} (\emph{$\LC^2$-equivalent}) if they satisfy the same formulas of the logic $\LL^2$ ($\LC^2$, respectively).
Now Immerman and Lander's theorem states that for all graphs $G$ (possible coloured and/or directed) and all vertices $v,w\in V(G)$, the vertices $v$ and $w$ have the same colour in the coarsest bi-stable colouring of $G$ if and only if they are $\LC^2$-equivalent. 
(Recall that bi-stable was defined in Section~\ref{ssec:extensions}.)
In particular, this implies that the $\LC^2$-equivalence classes of a graph can be computed in time $O((n+m)\log n)$, but not better (by a partition-refinement algorithm).

On plain undirected graphs, the logic $\LL^2$ is extremely weak. However, on coloured and/or directed graphs, the logic is quite interesting. The $\LL^2$-equivalence relation refines the bisimilarity relation. It is well known that the $\LL^2$-equivalence relation can be computed in time  $O((n+m)\log n)$ by a variant of the colour refinement algorithm. Our lower bounds can be extended to show that it cannot be computed faster by a partition-refinement algorithm.
 
\subsubsection{An Open Problem}

The key idea of the $O((n+m)\log n)$ partitioning algorithms is Hopcroft's idea of processing the smaller half. Hopcroft originally proposed this idea for the minimisation of deterministic finite automata. The algorithm proceeds by identifying equivalent states and then collapsing each equivalence class to a single new state. The partitioning problem (computing classes of equivalent states) is actually just the bisimilarity problem for finite automata, which may be viewed as edge-labelled transition systems. 

However, for DFA-minimisation we only need to compute the bisimilarity relation for \emph{deterministic} finite automata, that is, transition systems where each state has exactly one outgoing edge of each edge label. The systems in our lower bound proof are highly nondeterministic. Thus our lower bounds do not apply.

It remains a very interesting open problem whether similar lower bounds can be proved for DFA-minimisation, or whether DFA-minimisation is possible in linear time. Paige, Tarjan, and Bonic~\cite{paitarbon85} proved that this is possible for DFAs with a single-letter alphabets. To the best of our knowledge, the only known result in this direction is a family of examples due to Berstel and Carton~\cite{bercar04} (also see~\cite{casressci09,berboacar09}) showing that the $O(n\log n)$ bound for Hopcroft's original algorithm is tight.

\subsubsection*{Acknowledgements}
We thank Christof L\"oding for discussions and references regarding the lower bounds for bisimilarity and DFA-minimisation.

\bibliography{cr}

\end{document}